\documentclass[journal]{IEEEtran}
\usepackage[utf8]{inputenc} 
\usepackage[T1]{fontenc}
\usepackage{url}
\usepackage{ifthen}
\usepackage{cite}
\usepackage{amsmath,amsthm,amssymb, cite}
\usepackage{bm,bbm}
\usepackage{mathtools}
\usepackage[small]{caption}
\usepackage{tablefootnote}
\usepackage{multirow,color,amsfonts}
\usepackage{tabulary}
\usepackage{subfigure}
\usepackage{graphicx}
\usepackage{setspace}
\usepackage{enumerate}
\usepackage[]{algorithm2e}
\usepackage{comment}
\usepackage{cleveref}

\newtheorem{theo}{Theorem}
\newtheorem{defi}{Definition}

\newtheorem{prop}{Proposition}
\newtheorem{remark}{Remark}

\newtheorem{cor}{Corollary}

\DeclareMathOperator*{\etaM}{\eta_{\sf MatII}}
\DeclareMathOperator*{\etaGM}{\eta_{\sf GEC}}
\DeclareMathOperator*{\etaI}{\eta_{\sf Indep.}}

\DeclareMathOperator*{\HM}{H_{\sf MatII}}

\DeclareMathOperator*{\HI}{H_{\sf Indep.}}

\DeclareMathOperator*{\Hpi}{H_{\pi}}
\DeclareMathOperator*{\Hpii}{H_{\pi,i}}
\DeclareMathOperator*{\HpiI}{H_{\pi,\mathcal{I}}}

\DeclareMathOperator*{\PPP}{\sf{PPP}} 

\DeclareMathOperator*{\GEC}{\sf{GEC}} %$\gamma$-exclusion caching
\DeclareMathOperator*{\HEC}{\sf{HEC}} %hard exclusion caching 

\newcommand{\norm}[1]{\left\lVert#1\right\rVert}

\DeclareMathOperator*{\Rsp}{R_{\sf Sph}} 
\DeclareMathOperator*{\Rdd}{R_{\sf c}} 
\DeclareMathOperator*{\Rdds}{R_{\sf c}^2} 
\newcommand{\etapi}{\eta_{\pi}}

\newcommand{\Ex}[1]{\mathbb{E}\!\left[{#1} \right]}

\newcommand{\Var}[1]{\mathrm{Var}\!\Big[{#1} \Big]}
\newcommand{\Varpi}[1]{\mathrm{Var}_{\pi}\!\left[{#1} \right]}
\newcommand{\Cov}[1]{\mathrm{Cov}\!\left[{#1} \right]}

\begin{document}
\title{Spatial Concentration of Caching in Wireless Heterogeneous Networks}
\author{Derya~Malak, Muriel M\'{e}dard, and Jeffrey G. Andrews
\thanks{Malak is with the Dept. of ECSE, Rensselaer Polytechnic Institute, Troy, NY 12180 USA (email: malakd@rpi.edu).} 
\thanks{M\'{e}dard is with Research Laboratory of Electronics, MIT, Cambridge, MA 02139 USA (email: medard@mit.edu).}
\thanks{Andrews is with the Dept. of ECE at the University of Texas at Austin, TX 78712 USA (email: jandrews@ece.utexas.edu).}
\thanks{An early version of the paper appeared in Proc., IEEE ISIT 2019 \cite{malak2019spatial}. \hfill Manuscript last revised: {\today}.}
}

\maketitle

%%%%%%
\begin{abstract}
We propose a decentralized caching policy for wireless heterogeneous networks that makes content placement decisions based on pairwise interactions between cache nodes. We call our proposed scheme  \emph{$\gamma$-exclusion cache placement} ($\GEC$), where a parameter $\gamma$ controls an exclusion radius that discourages nearby caches from storing redundant content. $\GEC$ takes into account item popularity and the nodes' caching priorities and leverages negative dependence to relax the classic 0-1 knapsack problem to yield spatially balanced sampling across caches.  We show that $\GEC$  guarantees a better concentration (reduced variance) of the required cache storage size than the state of the art, and that the cache size constraints can be satisfied with high probability. Given a cache hit probability target, we compare the 95\% confidence intervals of the required cache sizes for three caching schemes: (i) independent placement, (ii) hard exclusion caching ($\HEC$), and (iii) the proposed $\GEC$ approach. For uniform spatial traffic, we demonstrate that $\GEC$ provides approximately a 3x and 2x reduction in required cache size over (i) and (ii), respectively. For non-uniform spatial traffic based on realistic peak-hour variations in urban scenarios, the gains are even greater. 
\end{abstract}


%%%
\section{Introduction}
\label{introduction}

Distributed caching is a powerful technique to minimize network latency \cite{Shanmugam2013}, and to enable spectral reuse and throughput gain in networks \cite{MaddahAli2013Journal}.  
The primary goal of cache placement in a wireless network is to maximize the cache hit probability, meaning the probability that a node in the network can obtain a desired item from a nearby cache within radio range.  This then eliminates the need for the network, for example a base station or other form of infrastructure, to fetch and broadcast the content.   The cache hit probability is affected by quantities such as the demand distribution, network topology, communication range, the cache storage size, and the topic of this paper, which is the policy for populating these small caches with content.

\subsection{Motivation and Objectives}
\label{motivation}
Our focus is to devise a decentralized caching policy by exploiting the spatial distribution of wireless caches. We can attain a target cache hit rate performance by trading off the local cache space with the spatial diversity of the caching. The aggregate cache size that a user has access to grows with the diversity across the network. This can ensure that different demands can be satisfied with multiple caches, yielding an improved caching performance with the diversity.  

We assume that the nodes are equipped with fixed size caches and are randomly and independently located. A baseline approach to placement is to independently populate the caches by capturing the popularity of the files. However, this approach may not perform well for the tail of the demand distribution when the number of files is large or when the spatial traffic is non-uniform. One way to improve on such a policy is to assume that a receiver can obtain a desired item as long as it is cached within its communication range, and devise a joint cache placement strategy across nodes so as to minimize the chance of cache misses. This will clearly outperform the independent placement policy, but it may be very complex to devise a spatial cache placement model that captures the joint interactions of all nodes. 

In our prior work \cite{Malak2016twc}, which provides a baseline for the current paper, we considered a pairwise interaction model between nodes to make placement decisions, referred to as hard exclusion caching ($\HEC$). 
$\HEC$ is based on Mat\'{e}rn hard-core process of type II (MatII) which is less regular than a lattice but more regular than a Poisson Point Process (PPP). More specifically in MatII, a node from the baseline process is retained -- that is, selected to cache an item -- only if there is no other node within a deterministic exclusion radius.   The MatII makes binary placement decisions: it caches an item with probability $1$ if the node separation is larger than the exclusion radius, else $0$.  Since each item has a different popularity, the exclusion radius is chosen to be a decreasing function of an item's popularity, which introduces pairwise interactions between caches.   While $\HEC$ captures node interactions, the pairwise placement decisions can eliminate caching a desired item within the communication range, rendering $\HEC$ sub-optimal. In addition, in heterogeneous networks where nodes have variable transmit powers and link qualities, 
intermittent node failures may occur, which further renders a fixed exclusion radius suboptimal. In addition, since the exclusion range is optimized for binary as opposed to probabilistic placement, $\HEC$ is not robust to changing or uncertain network conditions.

The above challenges motivate us to seek answers to the following questions:
(Q1) How can we devise a new decentralized caching policy such that the caching penalty is not as severe as for  
the $\HEC$ policy with a fixed exclusion region? (Q2) Is there a class of probabilistic functions that enables a desired caching configuration in light of the demand distribution and network model?  
(Q3) Can we provide a continuous relaxation of the above 0-1 ``knapsack problem"\footnote{The 0-1 knapsack problem is a combinatorial optimization problem which is NP-hard \cite{sahni1975approximate}. In the context of caching, this is equivalent to restricting the number of copies of each kind of item to zero or one.} as compared to $\HEC$ policy, while still spatially balancing content placement? (Q4) How much gain can such a policy attain in terms of the cache hit versus cache size (provisioning) tradeoff?

\subsection{Related Work}
The modeling and analysis of caching gain have attracted significant interest in recent years, and has been studied from many different perspectives. For example, some of the limits of caching via exploiting the tradeoff between the network bandwidth usage and local caches have been studied in \cite{MaddahAli2013Journal}, in which the content placement phase is carefully designed so that a single coded multicast transmission can simultaneously satisfy multiple content requests. 
From the viewpoint of capacity scaling laws, more specifically, the per-node capacity scaling law as the number of nodes $n$ grows to infinity, for the protocol model for wireless ad hoc networks in \cite{gupta2000capacity}, have been explored in the context of caching, e.g., in \cite{ji2015throughput}, \cite{jeon2017wireless}, and \cite{liu2017much}. In \cite{ji2015throughput}, throughput of a D2D one-hop caching network was shown to grow linearly with cache size, and be inversely proportional to the number of files, but independent of $n$, unlike the scaling behavior without caching. Rate-memory and storage-latency tradeoffs for caching have been studied in \cite{yu2018exact}, \cite{xu2017fundamental}. Caching has also been studied in the context of device-to-device (D2D) communications in \cite{Ji2014}, \cite{golrezaei2014scaling}, \cite{Malak2016_D2DCaching}, and interference management in \cite{naderializadeh2017fundamental}, \cite{maddah2015cache}, and in optimization of cloud and edge processing for radio access networks in \cite{park2016joint}, \cite{sengupta2016cloud}. In addition, D2D caching with limited cooperation has been envisioned to increase offloading gains under low energy consumption requirements \cite{deng2018benefits}, hybrid multi-tier caching has been studied in \cite{fan2020cache}, where macro base stations (BSs) tier deterministically caches the most popular contents and the helpers tier probabilistically caches the less popular contents to improve the content delivery probability and rate, and energy efficiency.  Stochastic geometry-based approaches have been exploited for cluster-centric caching \cite{afshang2016fundamentals} and cloud caching architectures \cite{azimi2020stochastic}, and cache aided throughput that captures the reliability of each transmission in modeling the cache hit probability has been proposed in \cite{chen2016probabilistic}. However, none of these approaches focus on the cache overshoot probability based on a distributed approach tailored for proximity-based caching schemes while attaining a desired cache hit performance.

Temporal caching models have been analyzed in \cite{Che2002} for popular cache replacement algorithms, e.g., least recently used (LRU), least-frequently used (LFU), first in first out (FIFO), and most recently used (MRU). Decentralized spatial LRU caching strategies  
have been developed in \cite{Giovanidis2016}. These combine the temporal and spatial aspects of caching, and approach the performance of centralized policies as the coverage increases. However, they are restricted to the LRU principle. A time-to-live 
policy with a stochastic capacity constraint and low variance has been proposed in \cite{BerHenCiuSch2015}. 
The BitTorrent protocol employs the rarest first and choke algorithms to promote diversity 
among peers, and foster reciprocation, 
respectively. These have been demonstrated in the context of peer-to-peer (P2P) file replication in the Internet \cite{LegKelMic2006}. 

There also exist studies focusing on decentralized and geographic content placement policies such as \cite{Shanmugam2013}, \cite{Blaszczyszyn2014}, \cite{IoannidisYeh2016}, \cite{IoannidisYeh2017}, \cite{Malak2016twc}. The main focus of these works is to maximize the average cache hit probability subject to an average cache size constraint, relaxing the integrality constraints of cache capacities.  
This optimization problem can be solved as a convex program. However, to the best of our knowledge, the related literature does not provide performance guarantees, leading to these three performance questions.   (P1) How far off is the optimized average cache size from the unique per-cache size requirement? (P2) How far off is the average cache hit rate from the attainable performance?  (P3) How consistent and predictable is the cache hit probability across the caches?   Cache size over-provisioning plays a critical role in addressing these issues, including for performance optimization in hybrid storage systems in computing which combine solid-state drives, with the hard disk drive technology \cite{oh2012caching}.  
Hence, it is essential to devise content placement techniques that ensure the concentration of the cache size as well as
provide a spatially balanced allocation and a lower cache hit variance across the nodes.

%%%%
\subsection{Key Aspects and Contributions} 
\label{contributions}

In this paper, we develop a decentralized spatial exclusion-based cache placement policy that exploits pairwise interactions in heterogeneous networks.  
The intuition behind exclusion-based  
caching models is to promote content diversity and reciprocation by ensuring caches storing the same item are never closer than some given distance.  While any exclusion-based approach can yield a negatively dependent thinning of the baseline PPP, the nature of the thinning depends on whether the exclusion is determined in a strict sense, or in a probabilistic manner. For example, the $\HEC$ policy in \cite{Malak2016twc} is strict: a subset of nodes is retained to cache an item based on binary decisions imposed by the Mat\'{e}rn’s deterministic exclusion range model.  

Our proposed policy generalizes the $\HEC$ policy in two ways. First, the exclusion range is not fixed for a given item type. Second, the placement decision is probabilistic rather than binary. 
This is achieved by making the exclusion range at each node variable and determined by a probability measure $\gamma$ for each item, which we refer to as $\gamma$-exclusion cache placement ($\GEC$) policy. The exclusion range for $\GEC$ will be gamma-distributed which is explained later in Sect. \ref{SSCC_policy}. $\GEC$ retains a subset of the nodes of the baseline PPP to cache an item via this probabilistic rule, as well as a combination of factors involving the nodes' caching priorities, and a continuous function $f(r,m,n)$ that makes the pairwise placement decisions corresponding to an item.  
This function decays in the distance $r$ between the node pairs, which in turn increases the retaining probability of the pairs, and is symmetric with respect to the node exclusion radii $m$, $n$ that are distributed according to $\gamma$ that captures the item popularity. 
In that sense, $\gamma$ is not a sufficient statistic for $\GEC$ as pairwise distances matter for making placement decisions. Under GEC, we observe probabilistic caching of most popular contents in a set of nodes sampled from a larger pool of nodes versus more deterministic caching of less popular items in a smaller subset of nodes with larger separation distance, which is unlike the proposed scheme of \cite{fan2020cache}.

Revising the $\HEC$ policy that induces a binary retaining of cache pairs for a content item, $\GEC$ enables the probability of placement decision to transition smoothly from $1$ to $0$ as a function of the pairwise distances and the exclusion radii of the nodes captured via $f(r,m,n)$.

$\GEC$ roots in spatially balanced sampling,  
which is motivated by the request arrivals. For example, in P2P networking, the actual demand distribution is not known by nodes, and the cache updates in each peer are triggered by the requests. Furthermore, the traffic density  
is in general not uniform across 
cellular networks, and the peak hour density can be approximated by a lognormal distribution \cite{zhou2015spatial}. 
Hence, instead of having a fixed exclusion range, it is desirable to have a variable exclusion range, depending on an item's popularity. The $\GEC$ policy comes to the fore by putting a mark distribution on the exclusion range of an item based on its popularity. The marks may correspond to the detection / coverage ranges or the transmit powers of the nodes in heterogeneous networks. Motivated by the performance guarantee issues P1-P3 and answering design questions Q1-Q4, via $\GEC$ we aim to provide a better tradeoff between the cache hit rate and the cache size violation.   
Our main contributions and $\GEC$'s use cases are: 
\begin{enumerate}[i.]
    \item $\GEC$ has desirable spatial and local properties. Via design of the pairwise placement function, $\GEC$ enables a spatially balanced sampling,  
    yields an improved concentration of the cache size eliminating the need of over-provisioning, and enhances multi-hop connectivity. 
    \item $\GEC$ yields a better cache hit rate versus provisioning tradeoff than the state of the art. $\GEC$ can provide about a 3x and 2x reduction in required cache storage size over independent placement \cite{Blaszczyszyn2014} and $\HEC$ placement \cite{Malak2016twc}, respectively. The achievable gains are demonstrated for both uniform and non-uniform spatial traffic types.
    \item $\GEC$ has connections with rarest first caching as it promotes the item diversity and reciprocation among the nodes. Hence, it can be well-suited for P2P applications.
    \item $\GEC$ is suited for enabling proximity-based applications (such as D2D and P2P), and offloading mobile users in heterogeneous networks because it allows a flexible exclusion range determined by the coverage or transmit powers, rendering GEC robust to heterogeneity. 
\end{enumerate}

\subsection{Organization and Notation}

In Sect. \ref{OptimizationFormulation} we formally define the cache hit probability as function of the placement configuration, and discuss the uniform and non-uniform spatial demand models we use. Sect. \ref{SSCC_policy} contains our main technical contributions where we detail the construction of $\gamma$-exclusion caching model, and provide its characterization to answer Q1-Q4 and P1-P3.  
In Sect. \ref{experiments} we numerically evaluate these models and contrast their cache overprovisioning performance, for spatially uniform and non-uniform settings and demonstrate the performance gains of $\GEC$ over the baseline models.

{\bf Notation.} 
Let $\Phi$ denote the mother point process (p.p.), and $\Phi_{th}$ be the child p.p. obtained via the thinning of $\Phi$. Let $\pi$ be a spatial caching policy that yields a set of child p.p.'s $\{\Phi_{th,i}\}_i$, where $\Phi_{th,i}$ is the set of retained points that cache item $i$. Let $A$ be a given bounded convex set in $\mathbb{R}^2$ containing the origin, and $rA$ be its dilation by the factor $r$. $\mathbbm{1}\{A\}$ is the indicator of event $A$. Let $B$ be a bounded Borel set. Let $\Phi(B)$ be the random number of points of the spatial p.p. $\Phi$ which lie in $B$. Any receiver can obtain the desired item $i$ if it is within a critical communication range $\Rdd$. Assume that $B=B_0(\Rdd)$, where $B_0(r)$ is a ball in $\mathbb{R}^2$ with radius $r$, centered at origin. The notation for the caching network is detailed in Table \ref{notationtable}.

\begin{table*}[h!]\footnotesize
\centering
\setlength{\extrarowheight}{0.5pt}
\begin{tabular}{| l | l |}
\hline
{\bf Definition} & {\bf Function} \\
\hline
Mother p.p.; Intensity of $\Phi$; Child (thinned) p.p. & $\Phi=\{x_k\}$; $\lambda$; $\Phi_{th}$\\
\hline
A spatial caching policy that yields a set of child p.p.'s $\{\Phi_{th,i}\}_i$ & $\pi$\\
\hline
\hline
Ball centered at origin with radius $r$; Communication range & $B_0(r)$; $\Rdd$ \\
\hline
Catalog size; Set of items; Cache storage size & $M$; $\{i\}_{i=1}^M$; $N$\\
\hline
A request path with length $|q|=K_q$; A request pair; Set of all requests & $q=\{q_1, q_2, \hdots , q_{K_q}\}$; $r=(i,q)$; $\mathcal{R}$\\
\hline
Request pmf in the uniform spatial demand setting & $p_r\sim {\rm Zipf}(\gamma_r)$ \\
\hline
Request traffic intensity pmf in the heterogeneous spatial demand setting & $\sim {\rm lognormal}(\mu^*,\sigma)$\\ 
\hline
\hline
Intensity of MatII as function of the exclusion radius $r_i$ & $\lambda_{\rm hcp}$\\
\hline
\hline
Bivariate marks of $\GEC$ & $\{(m_k^{(i)},v_k^{(i)})\}$\\
\hline 
Distribution of the mark exclusion radius $m^{(i)}$; Distribution of the weight of mark & $\mu^{(i)}=\Gamma(\alpha,\beta)$; $v_k^{(i)}\sim U[0,1]$\\
\hline
Pairwise deletion probability for two points with marks $m$ and $n$ separated by $r$  & $f(r, m, n)$\\
\hline
\end{tabular}
\caption{Notation.}
\label{notationtable}
\end{table*}

%%%
\section{How to Optimize the Caching Gain}
\label{OptimizationFormulation}
The locations of the nodes (caches) in the network are modeled by a homogeneous PPP $\Phi=\{x_k\}$ in $\mathbb{R}^2$ with intensity $\lambda$. 
The network serves content requests routed over $\Phi$. 
There are $M$ items in the network, each having the same size, and each node has the same cache storage size $N<M$. 
We consider both an isotropic demand model where the request distribution is uniform across the network and a non-isotropic model where the traffic intensity is heterogeneous. 

%%%%
\subsection{Demand Distribution Model}
\label{demand distribution}
Each user makes requests based on a Zipf popularity distribution over the set of the items. We assume that the global probability mass function (pmf) of such requests (demand) over the network is given by $p_r(i)=i^{-\gamma_r}\big/\sum\nolimits_{j=1}^M {j^{-\gamma_r}}$, where $\gamma_r$ determines the tilt of the Zipf distribution. We also let $\mathcal{I}\sim p_r$ be the random variable that models the demand. The demand profile is the Independent Reference Model (IRM), i.e., the standard synthetic traffic model in which the request distribution does not change over time \cite{traverso2013temporal}. In the uniform demand model, the intensity of the requests for item $i$, i.e., $\lambda_i$, follows
the global request distribution $p_r(i)$.

We next consider a more realistic demand model that varies geographically, inspired from the existing empirical models that demonstrate this
\cite{lee2014spatial}, \cite{ying2014characterizing}, \cite{deng2015ginibre}, \cite{zhou2015spatial}. 

%%%%%%
{\bf The peak hour traffic density in urban/rural networks.} 
In \cite{zhou2015spatial}, the geographical variation of traffic measurements were collected from commercial cellular networks, based on a large-scale measurement data set from commercial GPRS/EDGE cellular networks deployed in a major east province of China. 
They have demonstrated that the peak hour traffic density, i.e., the highest volume of the cell  traffic load per unit area in kilobytes per square kilometer, both in a dense urban area and a rural scenario can be modeled by a lognormal distribution. In addition, authors in \cite{ying2014characterizing} have optimized the BS deployments for both urban and rural scenarios, where the deployments should be Mat\'{e}rn cluster and Strauss hard-core processes, respectively.

{\bf Distribution of non-uniform demand.} Motivated by \cite{zhou2015spatial}, we characterize the peak hour traffic density by $D = e^S$ where $S\sim \mathcal{N}(\mu^*,\Sigma)$, where the entries of the covariance matrix satisfy the relation $\Sigma_{ij}=\sigma^2 e^{-d_{ij}/r}$ that is obtained exploiting the exponential variogram model where the correlation is a function of the pairwise distances between the receivers \cite{zhou2015spatial}. The variogram has an appropriate scaling $r$ that is  determined by range, where the peak hour traffic of the two positions beyond this range is almost uncorrelated \cite{zhou2015spatial}.  
Numerical values of the distribution parameters are given in Sect. \ref{experiments}. 
For the non-uniform traffic model, the effective intensity of the requests for item $i$ at $x\in \Phi$ is 
\begin{align}
\label{non_uniform_intensity}
\lambda_i^x= \frac{D_x}{\mathbb{E}[D]}\lambda_i 
= \frac{D_x}{\mathbb{E}[D]} p_r(i),
\end{align}
where $\lambda_i$ is the intensity of the requests for item $i$ in the uniform spatial traffic model, and $D_x$ is the traffic density at $x\in\Phi(B)$ and $\mathbb{E}[D]=\frac{1}{|\Phi(B)|}\sum_{y\in\Phi(B)} D_y$. Note that the isotropic demand result $p_r(i)$ in Sect. \ref{uniform_demand} is can be obtained when $D_x=\mathbb{E}[D]$. 

{\bf Tilting of demand.} We model the non-uniform spatial request distribution at $x\in\Phi(B)$ by exponential tilting of the isotropic demand random variable $\mathcal{I}$, which is parameterized by $\theta_x$:
\begin{align}
\label{tilted_demand}
p_r^x(i)=\frac{e^{i\theta_x}p_r(i)}{\mathbb{E}_{\mathcal{I}}[e^{\mathcal{I}\theta_x}]}=\frac{e^{i\theta_x}p_r(i)}{\sum\limits_{m=1}^M e^{m\theta_x} p_r(m)},\quad i=1,\hdots, M.
\end{align}
From (\ref{non_uniform_intensity}) and (\ref{tilted_demand}) we choose the tilting parameter $\theta_x$ such that
\begin{align}
\label{non_uniform_demand_v2}
e^{i\theta_x}=e^{i\log\left(\frac{D_x}{\mathbb{E}[D]}\right)},\quad i=1,\hdots, M.    
\end{align}
Note that if the original distribution for the intensity of the requests for item $i$ is Poisson$(\lambda_i)$, its tilted distribution parameterized by $\theta_x$ will also be a Poisson distribution with parameter $\lambda_i^x=e^{\theta_x}\lambda_i $. Hence, exponential tilting ensures the validity of (\ref{non_uniform_intensity}). 
When the original distribution is different from Poisson, $\theta_x$ has to be set to ensure that the local request distribution at $x\in\Phi(B)$ is valid. For example, if $p_r\sim {\rm Zipf}(\gamma_r)$, then $e^{i\theta_x}$ equals $(1/i)^{\gamma_{\theta_x}}$ for some $\gamma_{\theta_x}$ so that the tilted distribution is also a Zipf distribution with parameter $\gamma_r+\gamma_{\theta_x}$.

Another choice of the tilting parameter is based on a weighted version of exponential tilting such that the peak hour traffic density is higher for only a subset of items:
\begin{align}
\label{non_uniform_demand_v1}
e^{i\theta_x}=
\begin{cases}
\frac{\sum\limits_{x\in \Phi(B): D_x>\mathbb{E}[D]} D_x}{|\Phi(B)|\mathbb{E}[D]-\sum\limits_{x\in \Phi(B): D_x>\mathbb{E}[D]} D_x},\,\,\, i=1,\hdots, \frac{M}{2}\\
1,\,\,\, i=\frac{M}{2}+1,\hdots, M.
\end{cases}
\end{align}

{\bf Uniform versus non-uniform  demand.} Provided that non-uniform spatial traffic satisfies $\mathbb{E}_x[\lambda_i^x]=\lambda_i$, the average caching gain over all the users remains unchanged from the uniform model. However, non-uniform demand also affects the second order properties of caching, such as its spatial distribution and variance across nodes, which play a key role in the performance.

For simplicity of presentation, in the rest of this section, we focus on the uniform demand model and present the wireless caching gain function for different coverage models. The caching gain function can similarly be derived for the heterogeneous demand model using the local demand distribution, which then can be averaged over the network. Hence, we defer the discussion of the heterogeneous demand model to the numerical simulations section (see Sect. \ref{experiments}).

%%%%%%%
\subsection{Caching Gain Model under Different Wireless Coverage Scenarios}
\label{uniform_demand}

Each node $x\in\Phi$ is associated with the variables $z_{xi}=\mathbbm{1}\{i\in {\rm Cache}(x)\}$ that denote whether item $i$ is available in its cache or not. 
 Assuming that the network node have no caching capability. In that case, the requests are served by a base station which gives an upper bound on the expected caching cost (e.g., routing, download delay). 
 The goal of caching is by designing $Z=[z_{xi}]_{x\in\Phi,\, i=1,\hdots, M}$ to enable a maximum average reduction of the costs, i.e., maximize the caching gain, via content caching in the network. 
 Let $w_k$ be the cost associated with a requested item within the presence of $k$ nodes in the range of the user.  
 We assume that request forwarding costs are negligible, and requests and downloads are at a smaller timescale than the request arrivals, which is aligned with existing approaches, e.g., see \cite{IoannidisYeh2016}, \cite{IoannidisYeh2017}.  
Given these parameters, the wireless caching gain function or the cache hit probability can be defined as follows: 
\begin{align}
\label{generalcachecost}
F(Z)=\mathbb{E}_{\mathcal{I}}\left[\sum\limits_{k=1}^{\infty}{w_{k}\Big(1-\prod\limits_{k'=1}^k\big(1-z_{q_{k'} \mathcal{I}}\big)\Big)}\right],
\end{align}
where the expectation is with respect to the demand modeled via $\mathcal{I}$. Furthermore, when the demand is uniform, i.e., $\lambda_i=p_r(i)$, we can replace $\mathbb{E}_{\mathcal{I}}[\cdot]$ with $\sum\nolimits_{i}{\lambda_i [\cdot]}$. We next justify our choice and observations on the caching gain function in (\ref{generalcachecost}) through the following steps: 
\begin{enumerate}[i.]
\item $F(Z)$ is a monotonically nondecreasing submodular function of $Z$.
\item The product term in (\ref{generalcachecost}) satisfies the following relation:
\begin{align}
\label{product_observation}
f_{z_i}(1,\hdots,k)=\prod\limits_{k'=1}^k\big(1-z_{q_{k'} i}\big)\nonumber\\
=\begin{cases}
1,\,\, z_{q_{k'} i}=0,\,\, \forall k'\in\{1,\hdots,k\},\\
0,\,\, \text{otherwise}.
\end{cases}
\end{align}
\item If the nodes do not cache, i.e., $z_{xi}=0$ for all $x$ and $i$, then $F(Z)=0$. Hence, the cost of obtaining an item is $\sum\nolimits_{k=1}^{\infty}w_{k}$, i.e., the costs for all possible node configurations accumulate. 
\item  From (\ref{product_observation}) the term $\big(1-\prod\nolimits_{k'=1}^k\big(1-z_{q_{k'} \mathcal{I}}\big)\big)$ in (\ref{generalcachecost})  equals $0$ when the item is not cached in the subset of nodes $\{{q_{k'}} \}_{k'=1}^k$, and equals $1$ when the item is cached at least once in $\{{q_{k'}} \}_{k'=1}^k$, with the %empty product 
convention $\prod\nolimits_{k=1}^0 a_k=1$, yielding no caching gain versus a gain of $w_k$, respectively, given that there is a subset $\{q_1,\hdots, q_k\}\subseteq \Phi$ of $k$ nodes the user can obtain the item from. The caching gain is obtained via aggregating the costs $w_k$ of different user associations.
\item From (\ref{generalcachecost}) we observe a caching gain of $w_k$ if $k^*$ is the first $k$ such that $1-\prod\nolimits_{k'=1}^{k^*}\big(1-z_{q_{k'}i}\big)=1$ holds, i.e., the item is not available in any of $q_k$, $k=1,\hdots k^*-1$. This implies for $l< k^*$ that $1-\prod\nolimits_{k'=1}^{l}\big(1-z_{q_{k'}i}\big)=0$ and for $l\geq k^*$ that $1-\prod\nolimits_{k'=1}^l\big(1-z_{q_{k'}i}\big)=1$ because at least one node in $\{{q_{k'}} \}_{k'=1}^l$ in the range of the user has the desired item. As a result, the caching gain is $\sum\nolimits_{k=k^*}^{\infty} w_{k}$. Equivalently, the caching cost is given by %the difference between the costs of no caching and the gain via caching, i.e.,
$\sum\nolimits_{k=1}^{\infty} w_{k}-\sum\nolimits_{k=k^*}^{\infty} w_k=\sum\nolimits_{k=1}^{k^*-1} w_{k}$. 
\end{enumerate}

We next describe two coverage scenarios that follow from the caching gain function in (\ref{generalcachecost}).

{\bf Multi-hop coverage scenario.} In a multi-hop setting each node can use other intermediate nodes as relays to reach  some destination when radio range of each node is smaller than network coverage area, see e.g., \cite{IoannidisYeh2016}, \cite{IoannidisYeh2017}. Applications of multi-hop includes multi-path routing \cite{rossi2011caching}, and path replication algorithm \cite{cohen2002replication}, \cite{Che2002}, \cite{lv2002search} which is also being used in the context of caching. The replication algorithm is such that when an item traverses the reverse path towards a node that requested it, it is cached by every intermediate node encountered and the traditional cache eviction policies, such as LRU, LFU, FIFO, can be implemented when caches are full. 

A request is a pair $(i, q)$ that is determined by the item requested $i \in C$, and the path $q$ that the request traverses. We denote by $\mathcal{R}$ the set of all requests. Given a path $q$ that is a sequence $\{q_1, q_2, \hdots , q_{K_q}\}$ of nodes with length $|q|=K_q$ and the terminal node in the path is a designated source node for $i$, i.e., if $|q| = K_q$, then the item is always cached at $q_{K_q}$ and a $x \in q$, denote by $k_q(x)$ the position of $x$ in $q$, i.e., $k_q(x)$ equals $k \in \{1, \hdots , |q|\}$ such that $q_k = x$. In this case, the caching gain function in (\ref{generalcachecost}) can be given in this form:
\begin{align}
\label{MultiHopCachingGain}
F(Z)=\sum\limits_{(i,q)\in\mathcal{R}}{\lambda_{(i,q)}\sum\limits_{k=1}^{|q|-1}{w_{q_k}\Big(1-\prod\limits_{k'=1}^k\big(1-z_{q_{k'} i}\big)\Big)} },
\end{align}
where weight $w_{q_k} \geq 0$ associated with link $(q_k, q_{k+1})$ represents the cost of transferring an item across the link. If the item was not cached over the path, the cost of obtaining the item would be $\sum\nolimits_{k=1}^{|q|-1}w_{q_k}$. 
The cost of caching is accumulated over the path traversed by a request until the item is found, i.e., if $k^*$ is the first node index where item $i$ is stored on path $q$, then from step v. the caching gain for the item is $\sum\nolimits_{k=k^*}^{|q|-1}{w_{q_k}}$. 
We assume that the path originates from the request centered at origin, and $q_1$ is the nearest node, $q_2$ is the second nearest node, and so on. 
To determine $\{w_{q_k}\}$ we consider a path loss-based model and let  $w_{q_k}=\mathbb{E}[||q_k-q_{k+1}||^{\alpha}]$, where $\alpha$ is the path loss exponent. For $\Phi$, let $D_i$ be the distance of the $i$-th closest node. %to the origin. 
Then the joint distance distribution for $(D_1,D_2,\hdots, D_n)$ ordered in a nondecreasing way equals \cite{Ganti2012}
\begin{multline}
\label{distancedistribution}
f_{\vec{D}_{n}}(\vec{d}_{n})=f_{D_1,\hdots, D_n}(d_1,\hdots, d_n)\\
=(2\pi\lambda)^n \Big(\prod\limits_{i=1}^n d_i\Big) \exp(-\lambda\pi d_n^2),\quad 0\leq d_1\leq d_2\leq\hdots. 
\end{multline}
From the path loss model and (\ref{distancedistribution}), and exploiting the fact that $\Phi$ is isotropic, the weight satisfies
\begin{multline}
w_{q_k}=\frac{1}{2\pi}\int\limits_0^{2\pi}\int\limits_0^{\infty} \hdots \int\limits_{d_{k-1}}^{\infty}\int\limits_{d_k}^{\infty} (d_k^2+d_{k+1}^2 \nonumber\\
-2d_kd_{k+1}\cos(\theta))^{\alpha/2}  
f_{\vec{D}_{k+1}}(\vec{d}_{k+1}) {\rm d} \vec{d}_{k+1} {\rm d}\theta,\nonumber
\end{multline}
which can be plugged into (\ref{MultiHopCachingGain}) to determine the caching gain.

{\bf Boolean coverage scenario.} 
The Boolean coverage model captures the wireless networks in the noise-limited regime, and is also known as the SNR model. It has been used in the context of wireless caching in \cite{Giovanidis2016,Blaszczyszyn2014,Blaszczyszyn2014,Malak2016twc}. In wireless networks, the Boolean model can be realized for example via frequency reuse such that neighbouring nodes do not operate on the
same bandwidth. Because it is the noise-limited regime, the coverage cell of a typical node can be modeled using a ball $B$ of fixed radius (i.e., critical communication range $\Rdd$) centered at the node assuming that the random effects of the fading can be ignored (i.e., without channel variations). In Boolean scenario, each cache node has a possibly random area of coverage associated with it. Coverage cells of different nodes can overlap and a user at a random location may be covered by multiple nodes (all single hop away different from the multi-hop model) such that the user can obtain the desired item from which, or may not be covered at all. 

\begin{figure*}[t!]
\centering
    \includegraphics[width=0.8\textwidth]{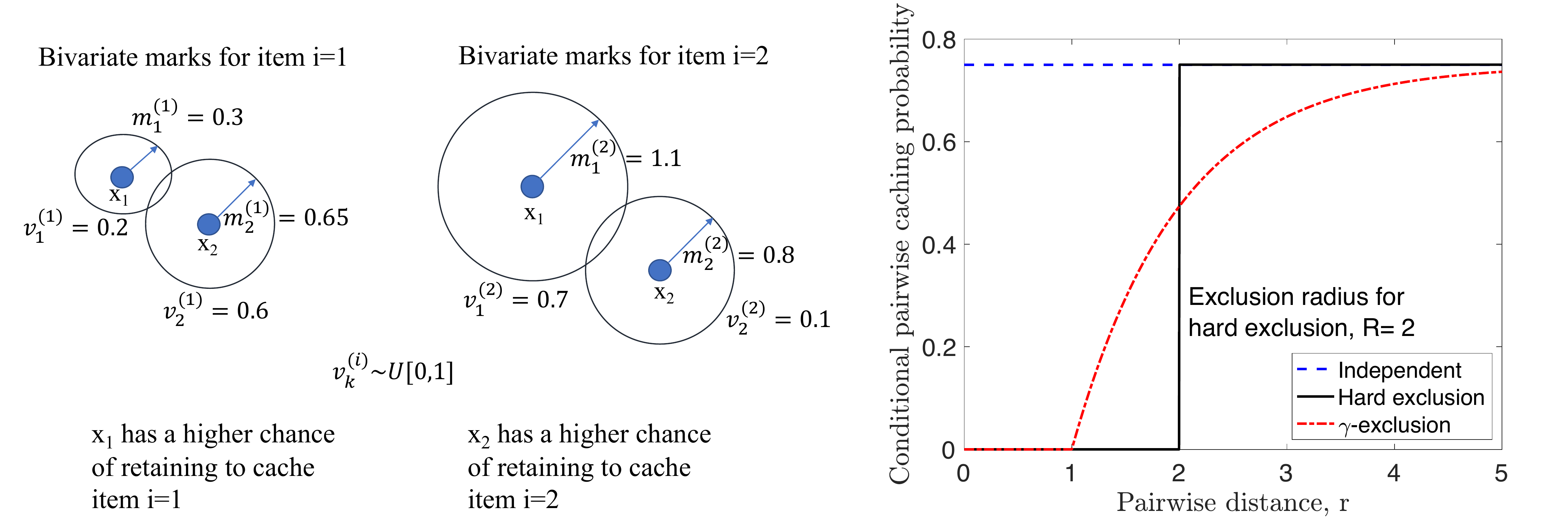}
\caption{\small{(Left) Bivariate marks for a node pair and 2 items. (Right) An example contrasting the probability of retaining a node of the baseline PPP for the independent caching, $\HEC$ and $\GEC$ policies as function of the pairwise distance. }}
\label{fig:gamma_function}
\end{figure*}

The caching gain in (\ref{generalcachecost}) for the Boolean coverage scenario is given as follows: 
\begin{align}
\label{BooleanCachingGain}
F(Z)= \sum\limits_{i\in\mathcal{R}}{\lambda_{i}\sum\limits_{k=0}^{\infty}w_k{\Big(1-\prod\limits_{k'=1}^k\big(1-z_{q_{k'} i}\big)\Big)} }, 
\end{align}
where $w_k=\mathbb{P}(\Phi(B)=k)$ is the probability that $k$ caches of the PPP $\Phi$ within $B$ cover the typical receiver. We assume that $\lambda_{i}=\lambda_{(i,q)}$ is the same across all paths. If there are $\Phi(B)=k^*$ and $k^*$ is the first index such that a cache has the desired item $i$, then from (\ref{generalcachecost}) and step $v$, the caching gain for item $i$ is given by the probability of having at least $k^*$ caches within the communication range, i.e., $\sum\nolimits_{k=k^*}^{\infty}{w_{k}}=\mathbb{P}(\Phi(B)\geq k^*)$ 
under the assumption that the worst case cost of caching is attained when $\Phi(B)=\infty$, i.e., none of the nodes cache the desired item. The caching gain increases in $w_k$ since the chance of finding the item locally increases.

We note that the multi-hop and Boolean coverage scenarios with the cache hit probabilities given in (\ref{MultiHopCachingGain}) and (\ref{BooleanCachingGain}), respectively, are indeed special cases of (\ref{generalcachecost}) and equivalent up to scaling of the costs (or weights). Thus, in the rest of the paper we focus on the Boolean coverage scenario.

\begin{figure*}[t!]
    \centering
    \includegraphics[width=0.9\textwidth]{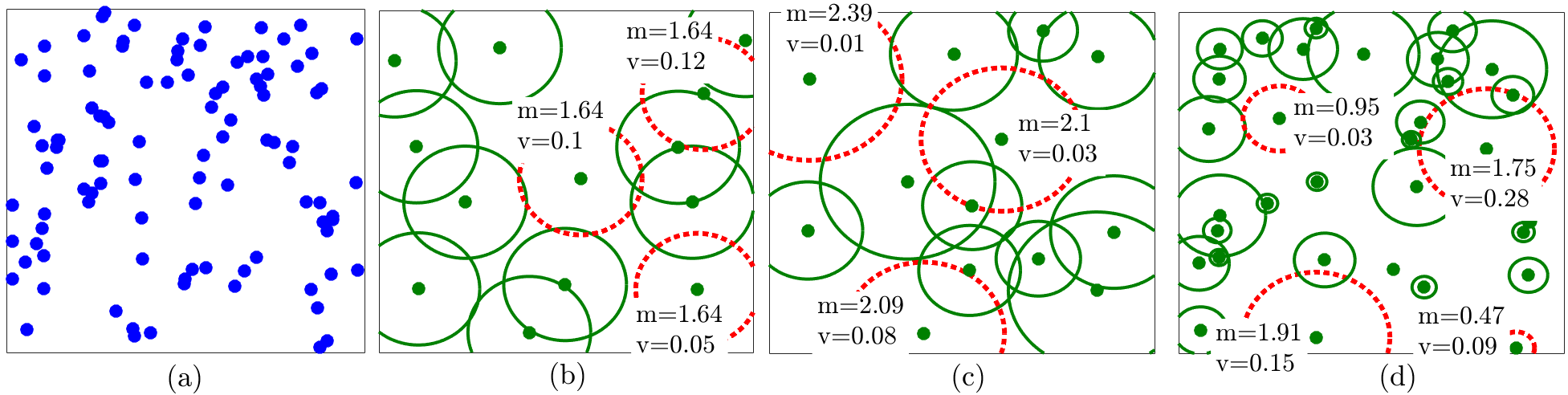}
    \caption{\small{$\GEC$ p.p. realizations: (a) Begin with a realization of PPP $\Phi$. Associate a gamma-distributed mark $m \sim \mu=\Gamma(\alpha,\beta)$ and a weight $v \sim  U[0, 1]$ to each node independently.  
    A node $x\in\Phi$ is selected based on (\ref{retainingprobability}), where $p_0=1$, $f(r,m,n)$ is given 
    in (\ref{fcfunction}) with $c=100$. The marks and weights for some of the retained points are shown in dotted circles. The retained p.p.'s for (b) $\beta=0$ (fixed mark radii), (c) $\beta=0.1$ (mark radii have low variance), and (d) $\beta=1$ (mark radii have high variance). As $\beta$ (the mark variance) increases, the packing is denser.}}
    \label{fig:softcoremodel}
\end{figure*}

%%%
\section{The $\gamma$-Exclusion Caching Model}
\label{SSCC_policy}
The 
$\gamma$-exclusion caching policy is constructed from the underlying PPP $\Phi$ 
with intensity $\lambda$, by removing certain nodes depending on the positions of the neighboring nodes, and on the marks and weights attached to them. It is based on a generalization of the Mat\'{e}rn II hard-core p.p. (MatII). 
The MatII process is obtained via the thinning of $\Phi$, where each node is equipped with different exclusion radii (deterministic marks) for distinct items according to their popularity \cite{Malak2016twc}. For this process, as the exclusion radius 
increases, the intensity $\lambda_{\rm hcp}(i)<\lambda$ corresponding to the set of nodes caching item $i$ decreases. The $\GEC$ policy is more general than MatII such that there is a distinct distribution that models the exclusion radius of each item. This can enable a softer version of thinning to ensure a more effective content placement.

For each item $i$, let $\tilde{\Phi}_i=\{(x_k,m_k^{(i)},v_k^{(i)})\}_{k}$ be a homogeneous independently marked PPP with intensity $\lambda$, and i.i.d. $\mathbb{R}^2$-valued marks, where $\Phi=\{x_k\}$, and $\{(m_k^{(i)},v_k^{(i)})\}$ is the random bivariate mark. The first component $m^{(i)}$ of the bivariate mark is referred to as mark, and has distribution $\mu^{(i)}$. The mark of item $i$, i.e., $m^{(i)}$, denotes its exclusion radius, and depends on its popularity in the network. If item $i$ is more popular than item $j$, then $m^{(i)}$ is stochastically dominated 
by $m^{(j)}$ \cite{marshall1979inequalities}. The marks $m_k^{(i)}$ are distributed according to $\mu^{(i)}$ for each $x_k$, and $i$. For example, we can associate a gamma-distributed mark to each node to model the variable wireless channel gains or compensate for the random effects of the fading ignored by the Boolean model. For the special case of MatII, i.e., when the marks are fixed, we optimized the exclusion radii in \cite{Malak2016twc}. 
The second component $v^{(i)}$ of the bivariate mark is weight, where $v^{(i)}$ serves as a weight of the thinning procedure, and is modeled by a random variable   
$\nu^{(i)}_{m^{(i)}}$ which might depend on $m^{(i)}$. In this paper, we assume that weights $v_k^{(i)}$ do not depend on the mark $m$. Instead, they are i.i.d. and uniformly distributed, i.e., $v_k^{(i)}$ are realizations of $\nu^{(i)}_{m^{(i)}}\sim \nu_m \sim U[0,1]$, for each  
$x_k$ and item $i$, meaning that the nodes are treated identically for purposes of thinning for caching. An example scenario for bivariate marks of 2 nodes $x_1,\,x_2\in \Phi$ is sketched in Fig. \ref{fig:gamma_function} (Left). 

Let $\Phi_{th,i}$ be a child p.p. that denotes the set of points that cache item $i$. The cache placement model is such that 
item $i$ is stored in cache $x_k\in \Phi$, i.e., $i\in {\rm Cache}(x_k)$, if and only if cache $x_k$ is kept as a point of $\Phi_{th,i}$. Equivalently, the caching variables satisfy 
\begin{align}
\label{GEC_cache_variable}
z_{x_ki}=\mathbbm{1}\{i\in {\rm Cache}(x_k)\}=\mathbbm{1}\{x_k\in\Phi_{th,i}\}. 
\end{align}
Node $x_k$ is retained as a point of $\Phi_{th,i}$ with probability $\mathbb{E}[z_{x_ki}]=p(x_k,m_k^{(i)},v_k^{(i)},\Phi)$. 

The number of items in cache $x_k$ is the sum of the individual items' indicator functions $C(x_k) = \sum\nolimits_i \mathbbm{1}\{i\in {\rm Cache}(x_k)\}$. The cache size constraint has to be satisfied on average, i.e.,
\begin{align}
\label{avgcachesize}
N=\mathbb{E}[C(x_k)]=\sum\limits_i p(x_k,m_k^{(i)},v_k^{(i)},\Phi),\quad x_k\in\Phi.
\end{align}
We next detail the dependent thinning procedure based on a pairwise deletion rule, and investigate the relationship between the mother p.p. $\Phi$ and the child p.p.'s $\Phi_{th,i}$, $i=\{1,\hdots, M\}$.

{\bf \em Dependent sampling of nodes for content placement.}  
From this part and onwards, for brevity of notation, we omit the index $i$, and consider the generic thinned process $\Phi_{th}$, derived from $\Phi$ by applying the following probabilistic dependent thinning rule. Assume that mark $m$ is distributed according to $\mu$, and $\vec{\delta}=\{m\}$ is the set of marks for all points in $\tilde{\Phi}$, where $m\sim \mu$ and $\bar{m}=\mathbb{E}_{m}[m]$. The marked point $(x, m, v) \in \tilde{\Phi}$ is retained as a point of $\Phi_{th}$ with probability
\begin{multline}
\label{retainingprobability}
p(x, m, v, \Phi) \\
= p_0 \prod\limits_{(y,n,u)\in\tilde{\Phi},\, y\neq x} [1 - \mathbbm{1}\{v \geq u\} f(||x - y||, m, n)]
\end{multline}
independently from deleting or retaining other points of $\Phi$. In other words, a node $x \in \Phi$ is retained to cache item $i$ with probability $p_0$, if it has the lowest weight among all the points within its exclusion range. In (\ref{retainingprobability}), $p_0 \in (0, 1]$, $f : [0, \infty[ \times \mathbb{R}^2 \to [0, 1]$ is a  mapping 
deterministic function satisfying $f(\cdot, m, n) = f(\cdot, n, m)$ for all $m,\, n \in \mathbb{R}$. This means that if two points with marks $m$ and $n$, and weights $v\geq u$ are a distance $r > 0$ apart, then the point with weight $v$  
is deleted by the other point with probability $f(r, m, n)$. Additionally, each surviving point is then again independently $p_0$-thinned. Conditional on retaining $(y,n,u)\in\tilde{\Phi}$ such that $y\neq x$, $1-f(||x - y||, m, n)$ represents the pairwise retaining probability of $x\in \Phi$ given that $v\geq u$.

The retaining policy in (\ref{retainingprobability}) would work under more general baseline p.p. models, such as the Binomial p.p., the Determinantal p.p., and the Ginibre p.p. and so on \cite{andrews2010primer}, while the resulting sampled processes will be different. For example, in \cite{deng2015ginibre}, the authors use the Ginibre p.p. to model random phenomena of BS positions with repulsion. If the baseline process is repulsive (where repulsion means that the node locations of a real deployment usually appear to form a more regular point pattern than the homogeneous PPP) or more deterministic like the binomial p.p. model, (\ref{retainingprobability}) will retain a larger portion of nodes compared to the PPP model versus for a clustered baseline process where (\ref{retainingprobability}) will result in a loss of a high fraction of points due to negative association between node pairs captured via the exclusion marks. Due to space limitations we leave the extension to different baseline models a future work.

The function $f(||x - y||, m, n)$ for $\GEC$ in (\ref{retainingprobability}) should be determined to ensure that (\ref{avgcachesize}) holds. Inspired from  \cite{teichmann2013generalizations}, we assume that $f$ is continuous and symmetric, and satisfies 
\begin{align}
\label{fcfunction}
    f(r,m,n)=\exp{(-c\lfloor r-m-n\rfloor_+)},\quad r\geq 0,
\end{align}
where $\lfloor x\rfloor_+ = \max\{0,x\}$. Fig. \ref{fig:gamma_function} (Right) clarifies the distinction of $\GEC$ and the independent caching and $\HEC$ policies, via indicating their pairwise retaining probabilities for  
$1-f(r,2,2)$.

Denote by $\GEC$$[\lambda,\mu,(\nu_m)_{m\in\mathbb{R}}, p_0, f]$ the distribution of $\Phi_{th}$. We next give its intensity, i.e., $$\lambda_{th}=\lambda \mathbb{E}[p(x, m, v, \Phi)].$$

\begin{theo}\label{theorem_intensity}\cite[Theorem 12]{teichmann2013generalizations}
The intensity of p.p. $\Phi_{th}\sim$$\GEC$$[\lambda,\mu,(\nu_m)_{m\in\mathbb{R}}, p_0, f]$ is given by 
\begin{multline}
\lambda_{th}=\lambda p_0 \int\nolimits_{\mathbb{R}} \int\nolimits_{\mathbb{R}} \exp\Big(-\lambda \int\nolimits_{\mathbb{R}} F_{\nu_n}(u)  \\
\int\nolimits_{\mathbb{R}^2}  f(||x ||, m, n){\rm d}x\ \mu({\rm d} n ) \Big)\, \nu_m({\rm d}u)\,\mu({\rm d}m),
\end{multline}
where $F_{\nu_m}(u)=\nu_m((-\infty,u])
$, $u\in\mathbb{R}$ 
is the cumulative distribution function (CDF) of $\nu_m$. 
\end{theo}

\begin{proof}
The probability generating functional (PGFL) \cite{Stoyan1996} of the PPP states for function $f(x)$ that $\mathbb{E}\left[\prod\nolimits_{x\in\Phi} f(x)\right] = \exp\big(-\lambda \int\nolimits_{\mathbb{R}^2} (1 - f(x)){\rm d}x\big)$. We can compute the intensity of $\Phi_{th}$ using
\begin{multline}
\lambda_{th}
=\lambda p_0 
\int\nolimits_{\mathbb{R}} \int\nolimits_{\mathbb{R}} \exp\Big(-\lambda \mathbb{E}\big[ \mathbbm{1}\{\nu_n\leq u\} \nonumber\\  
\int\nolimits_{\mathbb{R}^2}  f(||x ||, m, n){\rm d}x\ \big] \Big)\, \nu_m({\rm d}u)\,\mu({\rm d}m).\nonumber
\end{multline}
We obtain $\lambda_{th}$ using the PGFL and computing 
\begin{align}
\mathbb{E}[\mathbbm{1}\{\nu_n \leq u\}]=\int\nolimits_{\mathbb{R}} \mathbb{P}(\nu_n\leq u)\mu({\rm d}n)
=\int\nolimits_{\mathbb{R}} F_{\nu_n}(u) \mu({\rm d}n),\nonumber
\end{align}
along with incorporating the retaining probability in (\ref{retainingprobability}). 
\end{proof}

We next provide the intensity of MatII in terms of the marks as a special case of Theorem \ref{theorem_intensity}.

\begin{cor}\label{MatIIsimplified} Letting $f(\cdot, m, n)=1_{[0,m+n]}(\cdot)$, 
and $m=n=R/2$, the intensity of MatII equals 
\begin{align}
\lambda_{th}&=\frac{1-\exp\left(-\lambda     \pi R^2  \right)}{\pi R^2 }. 
\end{align}
\end{cor}

In Fig. \ref{fig:softcoremodel}, we illustrate different realizations of $\GEC$ $\Phi_{th}$ formed by thinning $\Phi$ for different mark radii where the mark distributions are gamma-distributed. A gamma-distributed random variable $X$ is characterized by a shape parameter $\alpha$ and a rate parameter $\beta$, and is denoted as $X\sim \Gamma (\alpha ,\beta ) \equiv \operatorname {Gamma} (\alpha ,\beta )$. The corresponding probability density function is $f(x;\alpha ,\beta )={\frac {\beta ^{\alpha }x^{\alpha -1}\exp{(-\beta x)}}{\Gamma (\alpha )}},\,\, {\text{ for }}x>0,\, \alpha ,\,\beta >0$, 
where its mean is $\alpha/\beta$, variance is $\alpha/\beta^2$, and  $\Gamma (\alpha )$ is the gamma function. As the mark variance increases, the packing becomes denser, which is desired for spatially balanced caching.  The gamma distribution is an exponential family. Exponential families arise in the context of the maximum-entropy distributions with specified means.

{\bf \em Spherical contact distribution function.} 
Our next goal in this section is to establish the relationship between the cache hit probability distribution, i.e., the distribution of $F(Z)$, to the spherical contact distribution function or empty space function, which we next formally define.

\begin{defi}{\bf Spherical contact distribution function \cite{Stoyan1996}.}
The spherical contact distribution function of the p.p. $\Xi$ is the conditional distribution function of the distance from a point chosen randomly outside $\Xi$ (i.e., $0$), to the nearest point of $\Xi$ given $0\notin \Xi$. It is given by 
\begin{align}
\label{sphericalCDF}
H(r) =\mathbb{P}(\Rsp\leq r\vert \Rsp>0), \quad r \geq 0,\nonumber\\
= 1 - \exp{\left( -\lambda \big[\mathbb{E}(\nu_2(\Xi_0 \oplus rA)) - \mathbb{E}(\nu_2(\Xi_0))\big] \right)},
\end{align}
where $\Rsp=\inf\{s: \Xi\cap sA\neq \emptyset\}$, $A=B_0(1)$ is the unit ball in $\mathbb{R}^2$ containing the origin, and $rA$ is the dilation of the set $A$ by the factor $r$, $\Xi_0$ is the typical grain of $\Xi$, $\mathbb{E}(\nu_2(\Xi_0))$ is the mean volume of the typical grain, $\Xi_0 \oplus rA$ denotes the Minkowski addition, and $\nu_2$ is the Lebesgue measure on $[\mathbb{R}^2,\mathcal{B}^2]$ where the $\sigma$-algebra $\mathcal{B}^2$ contains all the subsets of $\mathbb{R}^2$ that can be constructed from the open subsets by the basic set operations and by limits \cite{Stoyan1996}. 
\end{defi}

In Fig. \ref{fig:CDF} we illustrate the spherical contact distribution function for the Boolean model with random spherical grains  
\cite[Ch. 3.1]{BaccelliBook1} that model the coverage cells. For example, in the special case of PPP-MatII, the spherical grains have the same radius. We next exploit this notion to characterize the cache hit probability $F(Z)$ for different spatial content placement policies. 

\begin{figure}[t!]
    \centering
    \includegraphics[width=0.2\textwidth]{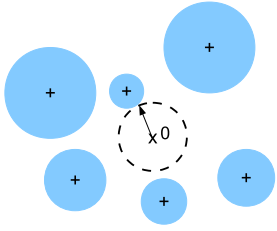}
    \caption{\small{The radius of the smallest sphere centered at $0$ intersecting the Boolean Model $\Phi_{th}$. The spherical contact distribution function is the conditional distribution function of the radius of the sphere, given $0\notin\Phi_{th}$ \cite[Ch. 3.1]{BaccelliBook1}.}}
    \label{fig:CDF}
\end{figure}

\begin{theo}\label{CachehitvsSCDF} The average cache hit probability of policy $\pi$ for a caching gain function $F(Z)$ that satisfies (\ref{generalcachecost}) is given as
\begin{align}
\label{avg_hit_rate_pi}
\mathbb{E}_{\pi}[F(Z)]=\mathbb{E}_{\mathcal{I}}[\HpiI(\Rdd)],
\end{align}
where $\HpiI(\Rdd)$ is the spherical contact distribution function of the thinned p.p. $\Phi_{th,\mathcal{I}}$ for $\mathcal{I}$.
\end{theo}

\begin{proof}
Let $B=B_0(\Rdd)$ and $\Phi_{th,i}(B)=\sum\nolimits_{x\in\Phi_{th,i}}1(x\in B)$ be the number of transmitters containing item $i$ within a circular region of radius $\Rdd$ around the origin. Then we have
\begin{align}
F(Z)=\sum\nolimits_{i}{p_r(i)\mathbbm{1}(\Phi_{th,i}(B)>0)}.\nonumber
\end{align}
The average cache hit probability is given by $\mathbb{E}[F(Z)]=\sum\nolimits_{i}{p_r(i)\mathbb{P}(\Phi_{th,i}(B)>0)}$, where defining $\Rsp=\inf\{s: \Phi_{th,i}(B_0(s)) \neq 0\}$, given $0\notin \Phi_{th,i}$ we have that
\begin{align}
\label{hit_cdf_equivalence}
    \mathbb{P}(\Phi_{th,i}(B)>0) = \mathbb{P}(\Rsp\leq \Rdd \vert \Rsp>0),
\end{align}
which is the spherical contact distribution function of $\Phi_{th,i}$ evaluated at $\Rdd$. 
\end{proof}

The variance of $F(Z)$ across the nodes satisfies the following additive relation
\begin{align}
\Varpi{F(Z)}=\sum\nolimits_{i} p^2_r(i) \Hpii(\Rdd) (1-\Hpii(\Rdd)) 
\end{align}
because 
$\Phi_{th,i}$ across $i=\{1,\hdots, M\}$ are independent of each other. Under the IRM and a Zipf popularity model, $\Varpi{F(Z)}$ decreases with increasing variance of marks when $\mathbb{E}[C(x)]$ is held constant \cite[Prop. 1]{BerHenCiuSch2015}. A spatially balanced sampling yields a low $\Varpi{F(Z)}$ as expected.

{\bf \em Migration to the child process via effective thinning.}  
Consider the pair $\Phi-\Phi_{th}$ of mother and child p.p.'s. The spherical contact distance denotes the distance between a typical point in $\Phi$ and its nearest neighbor from the thinned process $\Phi_{th}$. 
The spherical contact distribution function given in (\ref{sphericalCDF}) for the child process obtained via thinning of p.p. $\Phi$ can be rewritten as:
\begin{align}
\label{HRintegral}
\Hpi(R)=1 - \exp\Big( -\int\nolimits_{0}^{R}2\pi r\lambda \etapi(r,\delta)  {\rm d}r \Big), 
\end{align}
where $\etapi(r,\delta)$ is the conditional thinning Palm-probability, i.e., the probability of the point $x\in\Phi$ migrating to 
$\Phi_{th}$ under policy $\pi$, with a fixed (exclusion) radius $\delta$. It equals
\begin{align}
\label{ConditionalThinningPalm-Probability}
    \etapi(r,\delta)
    =\mathbb{P}(x\in \Phi_{th} \vert \Phi_{th}\cap B_{x_0}(r)=\emptyset, x_0\in \Phi).
\end{align}

We observe that the more effective the thinning 
policy $\pi$ is, the larger its conditional thinning Palm-probability $\etapi(r,\delta)$ in (\ref{ConditionalThinningPalm-Probability}) and the spherical contact distribution function $\Hpi(R)$ in (\ref{HRintegral}) are. From Theorem \ref{CachehitvsSCDF}, the expected caching gain $\mathbb{E}_{\pi}[F(Z)]$ is improved if $\pi$ is more effective.

We next formally provide the conditional thinning Palm-probability of the $\GEC$ policy.

\begin{prop} \label{CTPP_SSCC}
The conditional thinning Palm-probability for $\PPP$-$\GEC$ is given as 
\begin{multline}
    \etaGM(r,\vec{\delta})=\nonumber\\
    \int\nolimits_{\mathbb{R}}\int\nolimits_{0}^1
    \exp\Big(-u \lambda \int\nolimits_{\mathbb{R}} \int\nolimits_{\mathbb{R}^2} h(||x||,m,n)
    {\rm d}x \, \mu({\rm d}n) \Big)\, {\rm d}u\, \mu({\rm d}m),\nonumber
\end{multline}
where given radius marks $m,\,n$, $h(||x||,m,n)$ satisfies the relation 
\begin{align}
\int\nolimits_{\mathbb{R}^2} h(||x||,m,n)\,{\rm d}x=\pi(m+n)^2-l_2(r,n), 
\end{align}
where $l_2(r,\delta)$ is the area of the intersection of $B_{x_0}(r)$ and $B_{x}(\delta)$. It is given by
\begin{align}
\label{AreaIntersectingCircles}
    l_2(r,\delta)=\begin{cases}
    \pi r^2,\quad 0<r<\frac{\delta}{2}\\
    r^2\cos^{-1}\big(1-\frac{\delta^2}{2r^2}\big)+\delta^2\cos^{-1}\big(\frac{\delta}{2r}\big)\\
    -\frac{\delta}{2}\sqrt{4r^2-\delta^2},\quad r\geq \delta/2.
    \end{cases}
\end{align}

\end{prop}
 
\begin{proof}
The proof follows from generalizing \cite[Eq. (15)]{al2016nearest}.
\end{proof}

We next give special cases of Prop. \ref{CTPP_SSCC} corresponding to PPP-PPP and PPP-MatII thinnings, where significant simplifications on the conditional thinning Palm-probabilities and hence the contact distribution functions $\Hpi(R)$ in (\ref{HRintegral}) for these child point processes are possible. As a result, the average cache hit probabilities for such thinnings can be computed using Theorem \ref{CachehitvsSCDF}.

\begin{cor}
{\bf {\sf PPP}-{\sf PPP}.} The homogeneous PPP $\Phi_{th}$ with intensity $\lambda_{th}$ is obtained via independent thinning of $\Phi$, exploiting that $\delta=0$, where
\begin{align}
\label{PPPthinning}
    \etaI(r,\delta)=\int\nolimits_{0}^{1}\frac{\lambda_{th}}{\lambda}\exp{(-t\lambda \pi 0^2)}{\rm d}t=\frac{\lambda_{th}}{\lambda}.
\end{align}
Using (\ref{HRintegral}) and (\ref{PPPthinning}), the contact distribution function satisfies $\HI(R)=1-\exp{(-\pi \lambda_{th} R^2)}$. Incorporating $\HI(R)$ into Theorem \ref{CachehitvsSCDF} gives the average cache hit probability $\mathbb{E}_{\sf Indep.}[F(Z)]$.
\end{cor}

The contact distribution function for MatII, i.e., the CDF of the contact distance from a typical point in MatII to its nearest point in the same process, is derived in \cite[Eq. (10)]{al2016nearest}. 
The CDF of the PPP-MatII contact distance is larger than the MatII-MatII contact distance.

\begin{cor}{\bf {\sf PPP}-{\sf MatII}.} 
The conditional thinning Palm-probability for PPP-MatII is given as 
\begin{align}
\label{MatIIfixedexclusion}
    \etaM(r,\delta)&=\int\nolimits_{0}^{1}\exp{\left(-u\lambda (\pi\delta^2-l_2(r,\delta))\right)}{\rm d}u \nonumber\\
    &=\frac{1-\exp{\left(-\lambda (\pi\delta^2-l_2(r,\delta))\right)}}{\lambda (\pi\delta^2-l_2(r,\delta))}.
\end{align}
\end{cor}
Using (\ref{HRintegral}) and (\ref{MatIIfixedexclusion}), $\HM(R)$ can be computed, which in turn determines $\mathbb{E}_{\sf MatII}[F(Z)]$.

We next show that having a CDF 
on marks yields a more effective thinning than MatII does.
\begin{theo}\label{convexity}
The conditional thinning Palm-probabilities for $\GEC$ and MatII satisfy the relation
\begin{align}
\etaGM(r,\vec{\delta})\geq \etaM(r,\bar{m}),\nonumber 
\end{align}
where $\vec{\delta}=\{m\}$ is the set of marks in $\tilde{\Phi}$, with 
$\bar{m}=\mathbb{E}_{m}[m]$. 
\end{theo}

\begin{proof}
See Appendix \ref{App:convexity}.
\end{proof}

Exploiting Theorem \ref{convexity}, $\etaGM(r,\vec{\delta})$ is improved using a mixture of marks. The variable exclusion range model can suit to the case of cellular networks where the spatial demand is not uniform 
\cite{zhou2015spatial}. We numerically investigate the caching performance for non-uniform demand in Sect. \ref{experiments}.

To contrast %characterize 
the spatial performance and correlation of different caching models, we provide the functions describing the second-order behavior of the thinned p.p.'s in Appendix \ref{SecondOrderProperties}.

%%%%%
\section{Caching Optimization via Negative Association}
\label{cachingoptimization}

{\bf \em Spatial characterization via variances of point counts.} In this section, we contrast the performance of $\GEC$ with two baselines, in terms of first and second-order characterization of the cache hit probability. To that end, we devise a general negatively associated thinning model that embeds $\GEC$ as its special case which we detail in the sequel.

As the first baseline, we consider independent thinning. Then, the resulting p.p. for item $i$ is PPP with intensity $\lambda p_c(i)$, and from the Slivnyak-Mecke theorem \cite[Ch. 4.5]{Stoyan1996}, if $B=B_0(\Rdd)$, %is a bounded Borel set, 
it is satisfied that $\lambda K(\Rdd) =\mathbb{E}[\Phi_{ppp,i}(B)]=\Var{\Phi_{ppp,i}(B)}=\lambda p_c(i) \pi \Rdds$. Optimal caching pmf for PPP that maximizes the cache hit probability has been computed in \cite{Blaszczyszyn2014}, via a probabilistic allocation that places exactly $N$ items at each node $x\in \Phi_{ppp,i}$. This guarantees that the cache constraint is satisfied with equality, i.e., $\mathbb{E}[C(x)]=\sum\nolimits_{i} p_c(i)=N$ and ${\rm Var}[C(x)]=0$. 

As the second baseline model, we consider the MatII model with intensity $\lambda_{\rm hcp}(i)=\frac{1-\exp(-\lambda \pi r_i^2)}{\pi r_i^2}$, i.e., the caching probability for item $i$ is $p_c(i)=\frac{\lambda_{\rm hcp}(i)}{\lambda}$. In this model, since the placement across the nodes is independent across the items, at each node $x\in \Phi_{th,i}$, we have that $\mathbb{E}[C(x)]=\sum\nolimits_{i} p_c(i)=N$, and ${\rm Var}[C(x)]=\sum\nolimits_{i} p_c(i)(1-p_c(i))$.

Next result puts an upper bound on the variances of point-counts of $\Phi_{th,i}$ when $\Rdd$ is large.

\begin{prop}\label{spatial_var}\cite[Ch. 4]{Stoyan1996} 
The variances of point-counts of $\Phi_{th,i}$ for large $B$ is bounded as
\begin{align}
\label{SpatialVarianceMatIIhardcore}
\mathrm{Var}\big[\Phi_{th,i}(B)\big] \leq \lambda_{\rm hcp}(i)\pi \Rdds\exp{(-\lambda\pi r_i^2)}.
\end{align}
\end{prop}

\begin{proof}
See Appendix \ref{App:spatial_var}.
\end{proof}

From (\ref{SpatialVarianceMatIIhardcore}) spatial variance decays exponentially fast in $r_i^2$. This yields a more deterministic placement of rare items. Hence, MatII is spatially balanced, unlike independent placement.

The next propositions bound the average cache hit probability and its variance for MatII.

\begin{prop}\label{Avg_Hit_MHCP}
\cite{Malak2016twc} A lower bound on average cache hit rate for MatII is given as follows:
\begin{align}
F(Z)&\geq \sum\limits_{i=1}^{m_c}{p_r(i)[1-\exp(-\lambda_{\rm hcp}(i)\pi\Rdds)]}\nonumber\\
&+\sum\limits_{i=m_c+1}^{M}{p_r(i)\lambda_{\rm hcp}(i)\pi\Rdds},\nonumber
\end{align}
where $i\leq m_c: r_i<\Rdd$, i.e., the set $i= \{1,\hdots, m_c\}$ corresponds to the set of files within the communications range. 
An upper bound on average cache hit rate for MatII is given as follows:
\begin{align}
F(Z)\leq\sum\limits_{i=1}^{m_c}p_r(i)\Big[1-\exp(-\lambda_{\rm hcp}(i)\pi\Rdds) \nonumber\\
+\lambda^{-1}\int\nolimits_{B_0(\Rdd)}\rho_i^{(2)}(x){\rm d}x\Big]
+\sum\limits_{i=m_c+1}^{M}{p_r(i)\lambda_{\rm hcp}(i)\pi\Rdds}.\nonumber
\end{align}
\end{prop}

\begin{proof}
See Appendix \ref{App:Avg_Hit_MHCP}.
\end{proof}

\begin{prop}\label{MHCP}
The variance of the cache hit probability of MatII is upper bounded as 
\begin{align}
    \mathrm{Var}[F(Z)]\leq \sum\limits_{i=1}^{M}{p_r^2(i)\Big(\frac{1}{e}+\pi\lambda p_c^2(i)(\Rdds-r_i^2)_+\Big)}.
\end{align}
\end{prop}

\begin{proof}
See Appendix \ref{VarMHCII}. 
\end{proof}

We next formally introduce the concept of negative association and detail how to exploit it to devise caching models with desired second-order spatial properties. We then connect this notion to $\GEC$ to give a tight upper bound to cache over-utilization for $\GEC$.

{\bf \em Negative association and further implications in caching optimization.} 
In the rest of the section, we assume that $\{Z_i, \, 1\leq i\leq n\}$ be a feasible negatively associated sequence, and let $\{Z_i^*, 1\leq i\leq n\}$ be a sequence of independent random variables such that $Z_i^*$ and $Z_i$ have the same distribution for each $i=1,\hdots, n$. We next define negative association.

\begin{defi}\label{NAdef}{\bf Negative association \cite{DevPro1983}.}
A set of random variables $\{Z_i, \, 1\leq i\leq n\}$ is said to be negatively associated if for any two disjoint index sets $I, J\subseteq [n]$ and two functions $f,g$ both monotone increasing or both monotone decreasing, it holds
\begin{multline}
\label{NA_condition}
\Ex{f(Z_i: i\in I)\cdot g(Z_j: j\in J) }  \\ \leq \Ex{f(Z_i: i\in I)} \cdot \Ex{ g(Z_j: j\in J) }.
\end{multline}
\end{defi}

We note that for the caching gain function defined as in (\ref{generalcachecost}), $\GEC$ policy ensures that the caching variables $z_{x_ki}$, $x_k\in \Phi$, are negatively associated for all items $i$ in (\ref{GEC_cache_variable}) as (\ref{NA_condition}) holds. The next result bounds the cost of negatively associated caching that is a result of a coupling.

\begin{theo}\label{NAcostbound}{\bf Bounding the cost \cite[Theorem 1]{Shao2000}.}
For any convex function $f$ on $\mathbb{R}^1$
\begin{align}
\label{boundedfcost}
\mathbb{E}\Big[f\Big(\sum\limits_{i=1}^n Z_i\Big)\Big] \leq \mathbb{E}\Big[f\Big(\sum\limits_{i=1}^n Z_i^*\Big)\Big]
\end{align}
if the expectation on the right hand side of (\ref{boundedfcost}) exists. If $f$ is also non-decreasing, then 
\begin{align}
\label{maxboundedfcost}
\mathbb{E}\Big[f\Big(\max\limits_{1\leq k\leq n}\sum\limits_{i=1}^k Z_i\Big)\Big] \leq \mathbb{E}\Big[f\Big(\max\limits_{1\leq k\leq n}\sum\limits_{i=1}^k Z_i^*\Big)\Big]
\end{align}
provided that the expectation on the right hand side of (\ref{maxboundedfcost}) exists.
\end{theo}

The following observation on the caching gain function defined in (\ref{generalcachecost}) is immediate.
\begin{prop}\label{NA}
$F(Z)$ is convex if $\{z_{q_{k'} i}\}_{k'\in\{1,\hdots,k\}}$'s are negatively associated for all items $i$ \cite{Dubhashi1996p2}. 
\end{prop}

\begin{proof}
Using (\ref{product_observation}), we provide the first-order characterization of $f_{z_i}(1,\hdots, k)$ as 
$\mathbb{E}\Big[f_{z_i}(1,\hdots, k)\Big]
=\mathbb{P}(\Phi_{th,i}(B)=0)$. Exploiting the caching gain function in (\ref{generalcachecost}), observe that the following holds: 
\begin{align}
\label{NA_avg_cache_rate_bound}
&\mathbb{E}[F(Z)]=\mathbb{E}_{\mathcal{I}}\Big[\sum\limits_{k=0}^{\infty}w_k{\left(1-\mathbb{E}\left[f_{z_{\mathcal{I}}}(1,\hdots, k)\right]\right)} \Big]\\
&\overset{(a)}{\geq} \mathbb{E}_{\mathcal{I}}\Big[\sum\limits_{k=0}^{\infty}w_k{\Big(1-\prod\limits_{k'=1}^k\left(1-\mathbb{E}[z_{q_{k'} \mathcal{I}}]\right)\Big)} \Big]=F(\mathbb{E}[Z]), \nonumber
\end{align}
where $(a)$ is due to $\mathbb{E}\Big[f_{z_i}(1,\hdots, k)\Big]\leq\prod\nolimits_{k'=1}^k\big(1-\mathbb{E}[z_{q_{k'} i}]\big)$ as $z_{q_{k'} i}$'s are negatively associated.  
\end{proof}

From Prop. \ref{NA}, $\mathbb{E}[F(Z)]\geq F(\mathbb{E}[Z])$. The expected cache hit probability obtained via negatively associated placement upper bounds the independent placement solution with probabilities $\mathbb{E}[z_{q_{k'} i}]$. Negative association has desirable properties in terms of sampling and concentration. Some important results that hold for independent variables, e.g., the Chernoff-Hoeffding bounds, and the Kolmogorov's inequality \cite{Dubhashi1996p2}, \cite{wajc2017negative}, also hold for negatively associated variables. Hence, Prop. \ref{NA} justifies that negative association across caches is desirable.

\begin{prop}
\label{NegDepPlacement}
The average cache hit probabilities of negatively associated and independent caching policies satisfy the relation $\Ex{F(Z)}\geq\Ex{F(Z^*)}$.
\end{prop}
\begin{proof} 
If $\{Z_i, \, 1\leq i\leq n\}$ satisfy the negative association condition, then for any non-decreasing (or non-increasing) functions $g_j$, $j\in [n]$ \cite[Lemma 2]{Dubhashi1996p2}, $\mathbb{E}\Big[\prod\nolimits_{j\in[n]} g_j(Z_j)\Big] \leq \prod\nolimits_{j\in[n]} \Ex{g_j(Z_j)}$. The proof follows from contrasting the caching gain function in (\ref{generalcachecost}) of $\{Z_i\}$ versus  
$\{Z_i^*\}$.
\end{proof}

The following corollary from \cite{wajc2017negative} demonstrates that negatively associated placement policies have lower variance across nodes, hence are more stable than independent placement policies.
\begin{cor}\label{subadditivevariance} 
 \cite{wajc2017negative}. The variables $\{Z_i, \, 1\leq i\leq n\}$ satisfy $\mathrm{Var}\big[\sum\limits_{i=1}^n Z_i\big] \leq \sum\limits_{i=1}^n \mathrm{Var}[Z_i]$.
\end{cor}
From Cor. \ref{subadditivevariance}, we infer $\mathrm{Var}\big[\sum\nolimits_{k'=1}^k z_{q_{k'} i}\big] \leq \sum\nolimits_{k'=1}^k \mathrm{Var}\big[z_{q_{k'} i}\big]$. Since variance is sub-additive for negatively associated variables, upper bounding the number of times a network caches the same item in a Boolean coverage setting or in a path of the multi-hop scenario, implies less redundancy. Hence, negative association provides less variability within the coverage region, i.e., a more balanced allocation in a given path or in a Boolean coverage region such that an item is not cached in a redundant manner. For example, a more balanced allocation is given in \cite{AzaBroKarUpf1999}.

\begin{remark}
{\bf Use of negative association in caching  
\cite{wajc2017negative}.} When $\{Z_i\}$ is a feasible negatively associated solution that minimizes the convex function $f\Big(\sum\limits_{i=1}^n Z_i\Big)$, then the expected value of this solution is upper bounded by $\mathbb{E}\Big[f\Big(\sum\limits_{i=1}^n Z_i^*\Big)\Big]$, which is easier to compute as $Z_i^*$ are independent. 
\end{remark}

\begin{figure*}[t!]
    \centering
    \includegraphics[width=0.4\textwidth]{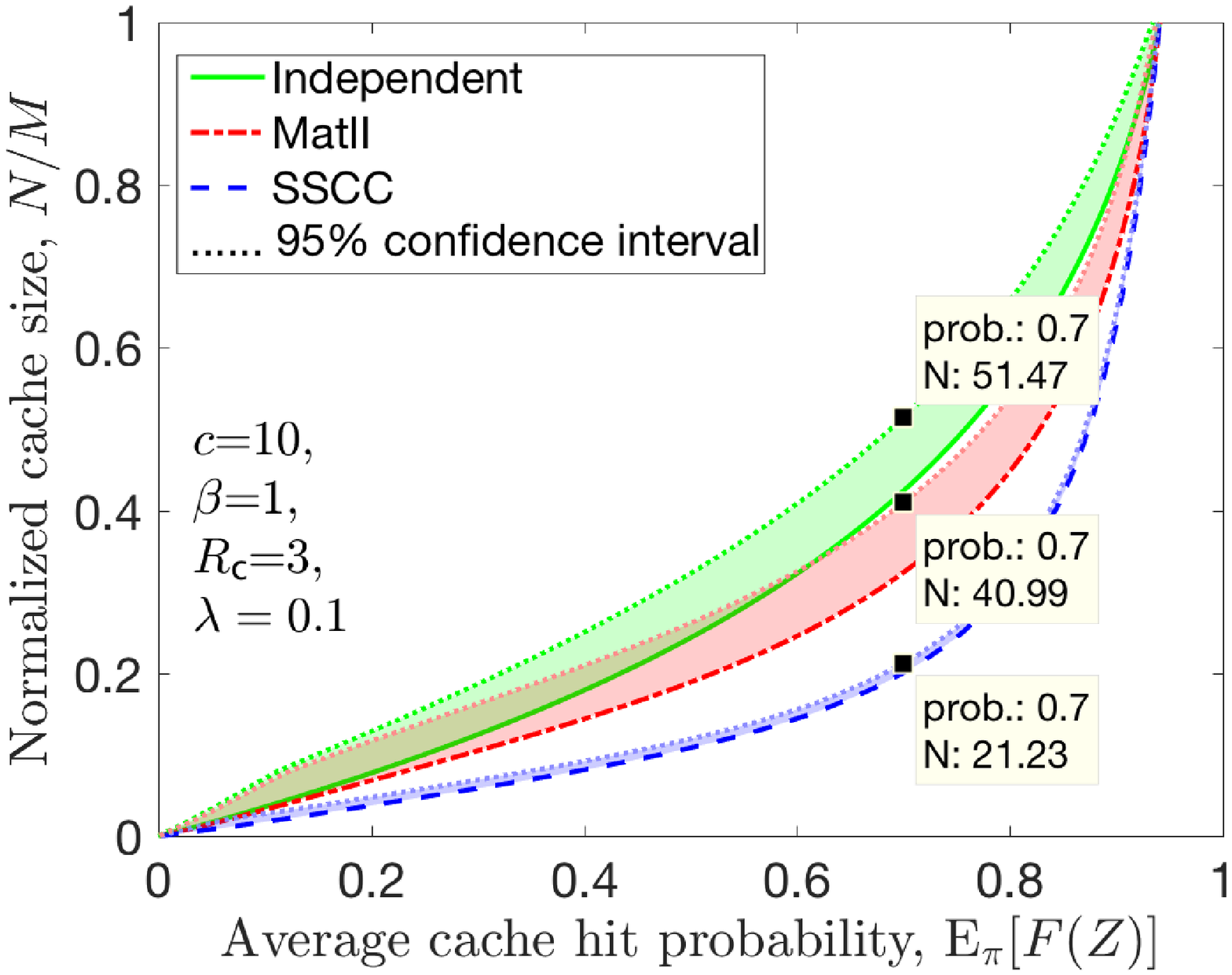}
    \includegraphics[width=0.4\textwidth]{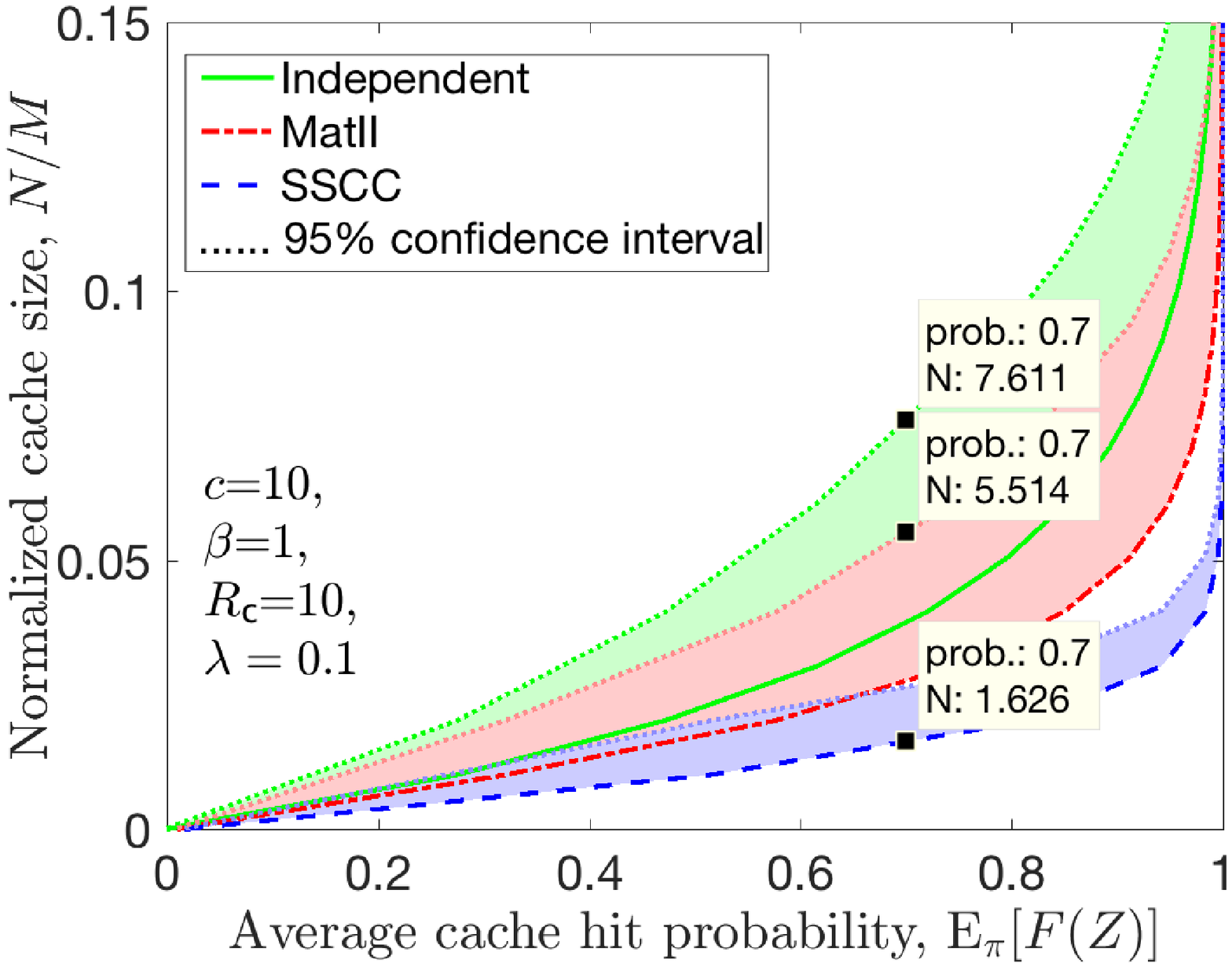}
    \caption{\small{The normalized cache size $N/M$ versus average cache hit rate for different placement policies. (Left) $\Rdd=3$, (Right) $\Rdd=10$. The intensity for the baseline process is $\lambda=0.1$. The request process is isotropic. The parameter $c$ for function $f(r,m,n)$ in (\ref{fcfunction}) of $\GEC$ is selected as indicated. The gamma-distributed mark distribution satisfies $\Gamma (\alpha=\mathbb{E}_{\mathcal{I}}(R_{\mathcal{I}})/\beta,\beta )$ with $\beta=1$, where $R_{\mathcal{I}}=\{R_i\}_i$ are the exclusion radii optimized for MatII. Note that the dotted curves correspond to the $95\%$ confidence intervals. The remaining curves correspond to the average of the normalized cache size $N/M$ for different content placement techniques. On the figures, we have marked the 95th percentile of the cache sizes $N$ required to achieve an average cache hit probability of $0.7$.}}
    \label{fig:performance1}
\end{figure*}

\begin{prop}\label{var_neg_assoc}
The variances of the cache hit probabilities of negatively associated and independent caching policies satisfy the relation $\mathrm{Var}[F(Z)]\leq \mathrm{Var}[F(Z^*)]$.
\end{prop}

\begin{proof}
See Appendix \ref{App:var_neg_assoc}.
\end{proof}

{\bf \em A concave approximation.} The caching gain function $F(Z)$ in (\ref{generalcachecost}) is not concave. Following the approximation technique in \cite{AgeSvi2004}, a concave approximation of $F(Z)$ is given as
\begin{align}
L(Z)=\sum\limits_{(i,q)\in\mathcal{R}}{\lambda_{(i,q)}\sum\limits_{k=1}^{|q|-1}{w_{q_{k+1}q_k}(i,q)\min\Big\{1,\sum\limits_{l=1}^k z_{q_l i}\Big\}} }.\nonumber
\end{align}
From Theorem \ref{NAcostbound}, we infer that the approximate caching gain function satisfies $L(Z^*)\leq L(Z)$. 
We also note that  
the solution of the concave relaxation of the expected caching gain $L(Z^*)$ will be within $1-1/e$ factor from the optimal expected caching gain, through pipage rounding \cite{AgeSvi2004}. Furthermore, from Prop. \ref{NA}, we have $\mathbb{E}[F(Z)]\geq F(Z^*)$ where $Z^*=\mathbb{E}[Z]$ which implies that for any $Z$, there exist the set of variables $Z^*$ such that $\mathbb{E}[F(Z)]$ is guaranteed to be at least within a $1-1/e$ approximation from the optimal caching gain solution that can be achieved via $Z^*$.

Although there might be multiple policies that maximize $\mathbb{E}[F(X)]$, a negatively associated allocation can have a smaller variance than the two baselines (independent and MatII models). The variance of a caching policy can help characterize the loss of the algorithm in terms of cache over-provisioning. We next investigate how much over-provisioning we require for $\GEC$.

{\bf \em Cache over-utilization.} The cache placement requires $C(x)=\sum\nolimits_{i}\mathbbm{1}_{x\in\Phi_{th,i}}\leq N$, for all $x\in\Phi$, where $N$ is finite. This constraint is satisfied on average, i.e., $N=\mathbb{E}[C(x)]=\sum_i p(x,m_i,v,\Phi)$, $x\in\Phi$. However, the p.p.'s $\{\Phi_{th,i}\}_i$, $i=1,\hdots, M$ might overlap. We need to make sure that the cache capacities are not over-utilized. The cache violation probability is the probability that an enforced average cache capacity constraint $N$ is violated by an amount $\epsilon$, i.e., $\mathbb{P}(C(x)>N+\epsilon)$. Hence, the intersection of the sampled processes, i.e., $\cap_i \Phi_{th,i}$, should not include any $x\in\Phi$ more than $N$ times with high probability. We next provide an upper bound for the violation probability of the cache size for $\GEC$, exploiting the Chernoff bound.

\begin{prop} 
The cache violation probability of $\GEC$ is upper bounded as 
\begin{align}
    \mathbb{P}(C(x)>N+\epsilon)\leq \exp{\Big(\epsilon - (N+\epsilon)\log\Big(1+\frac{\epsilon}{N }\Big)\Big)},\nonumber
\end{align}
where $\epsilon>0$ is an arbitrarily small constant.
\end{prop}

\begin{proof} We can rewrite the cache violation probability as
\begin{align}
    &\mathbb{P}(C(x)>C)
    =\mathbb{P}\Big(\sum\limits_{i}\mathbbm{1}_{x\in\Phi_{th,i}}>C\Big)\nonumber\\
    &=\sum\limits_{A\subseteq S\,: |A|>C}\mathbb{P}(x\in\Phi_{th,i},\,i\in A,\, x\notin\Phi_{th,j},\,j\in A^{\mathsf{c}})\nonumber\\
    &\overset{(a)}{=}\sum\limits_{A\subseteq S\,: |A|>C}\,\,\prod\limits_{i\in A}\mathbb{P}(x\in\Phi_{th,i})\prod\limits_{j\in A^{\mathsf{c}}}(1-\mathbb{P}(x\in\Phi_{th,j}))\nonumber\\
    &\overset{(b)}{\leq} \exp{\Big(C-\sum\limits_i p(x,m_i,v,\Phi) - C\log\Big(\frac{C}{\sum\limits_i p(x,m_i,v,\Phi) }\Big)\Big)},\nonumber
\end{align}
where $C=N+\epsilon$ for $N=\mathbb{E}[C(x)]$ and $\epsilon>0$ is an arbitrarily small constant, $S$ is the set of all subsets of $\{1,\hdots, M\}$, and $A$ is a subset of $S$ and $A^{\mathsf{c}}$ is its complement, and $(a)$ is due to the independent sampling of points for each $i$, i.e., $\mathbb{P}(x\in\Phi_{th,i},\,x\in\Phi_{th,j})=\mathbb{P}(x\in\Phi_{th,i})\mathbb{P}(x\in\Phi_{th,j})$ for all $i\neq j$, and $(b)$ follows from employing the Chernoff bound.
\end{proof}

As the Chernoff bound does not capture the second-order characteristics of the sampling, we next provide another bound based on the Bernstein inequalities.

\begin{figure*}[t!]
    \centering
    \includegraphics[width=0.4\textwidth]{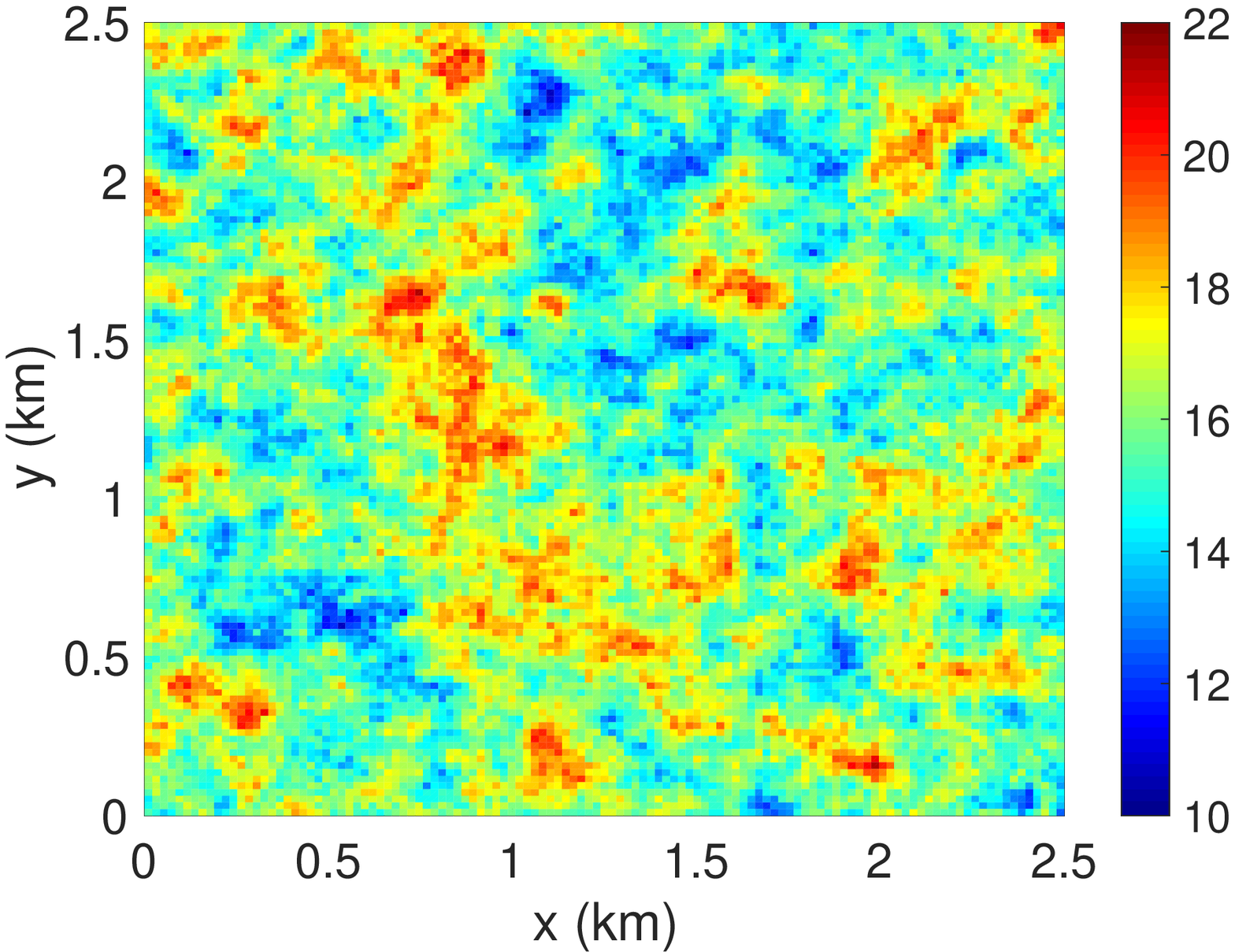}
    \includegraphics[width=0.4\textwidth]{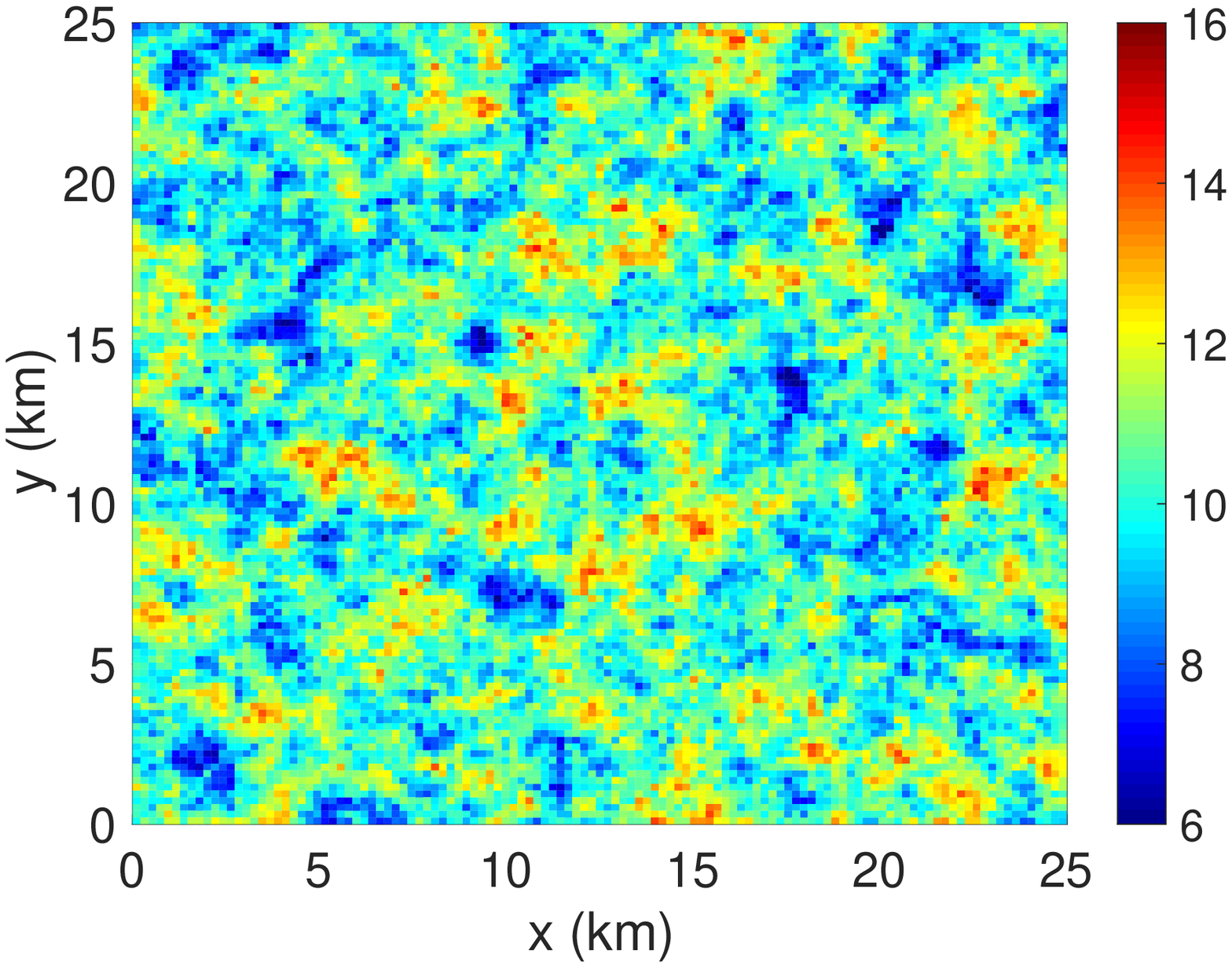}
    \caption{\small{The logarithm of the peak hour traffic density for the non-isotropic request model, which is generated based on actual traffic demand in \cite{zhou2015spatial}. (Left) Urban area. (Right) Rural Area.}}
    \label{fig:nonisotropic_request}
\end{figure*}

\begin{prop}
The cache violation probability of $\GEC$ is upper bounded as
\begin{align}
\label{Bernstein}
\mathbb{P}(C(x)>C) \leq \exp\Big( -\frac{(C-N)^2}{{\rm Var}[C(x)]+\frac{1}{3}(C-N)} \Big),
\end{align}
where ${\rm Var}[C(x)]=\sum\nolimits_{i=1}^M{\rm Var}[z_{xi}]$ since the placement is independent across items, where for each $i\in\{1,\hdots, M\}$
\begin{align}
&{\rm Var}[z_{xi}]=\mathbb{E}[z_{xi}^2]-\mathbb{E}[z_{xi}]^2\nonumber\\
&=p(x,m_i,v,\Phi)(1-p(x,m_i,v,\Phi)), \,\, x\in\Phi.\nonumber
\end{align}

\end{prop}
\begin{proof}
It follows from employing Bernstein inequality as $\mathbbm{1}_{x\in\Phi_{th,i}}$ are independent  
across $i$.
\end{proof}

As ${\rm Var}[C(x)]$ drops, the upper bound in (\ref{Bernstein}) drops. Hence, the cache violation probability of $\GEC$ is negligible if the cache placement strategy has a very low-variance.  
While it is nontrivial to design 
negatively associated placement techniques to satisfy the cache size constraint with probability $1$, in Sect. \ref{experiments}, we demonstrate that for $\GEC$ the cache violation probability can be made negligibly small, and contrast its performance with independent caching and $\HEC$.

\begin{figure*}[t!]
    \centering
    \includegraphics[width=0.4\textwidth]{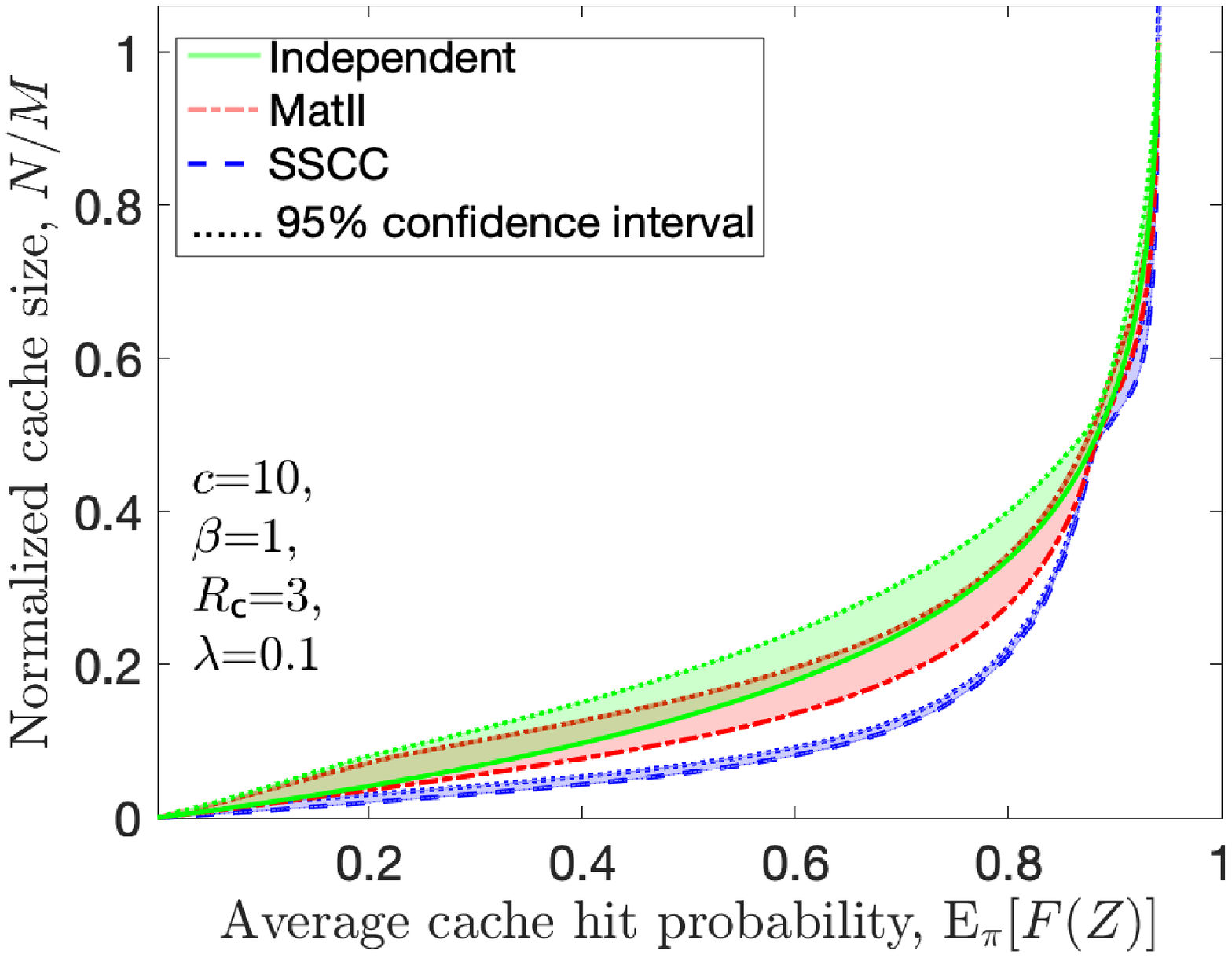}
    \includegraphics[width=0.4\textwidth]{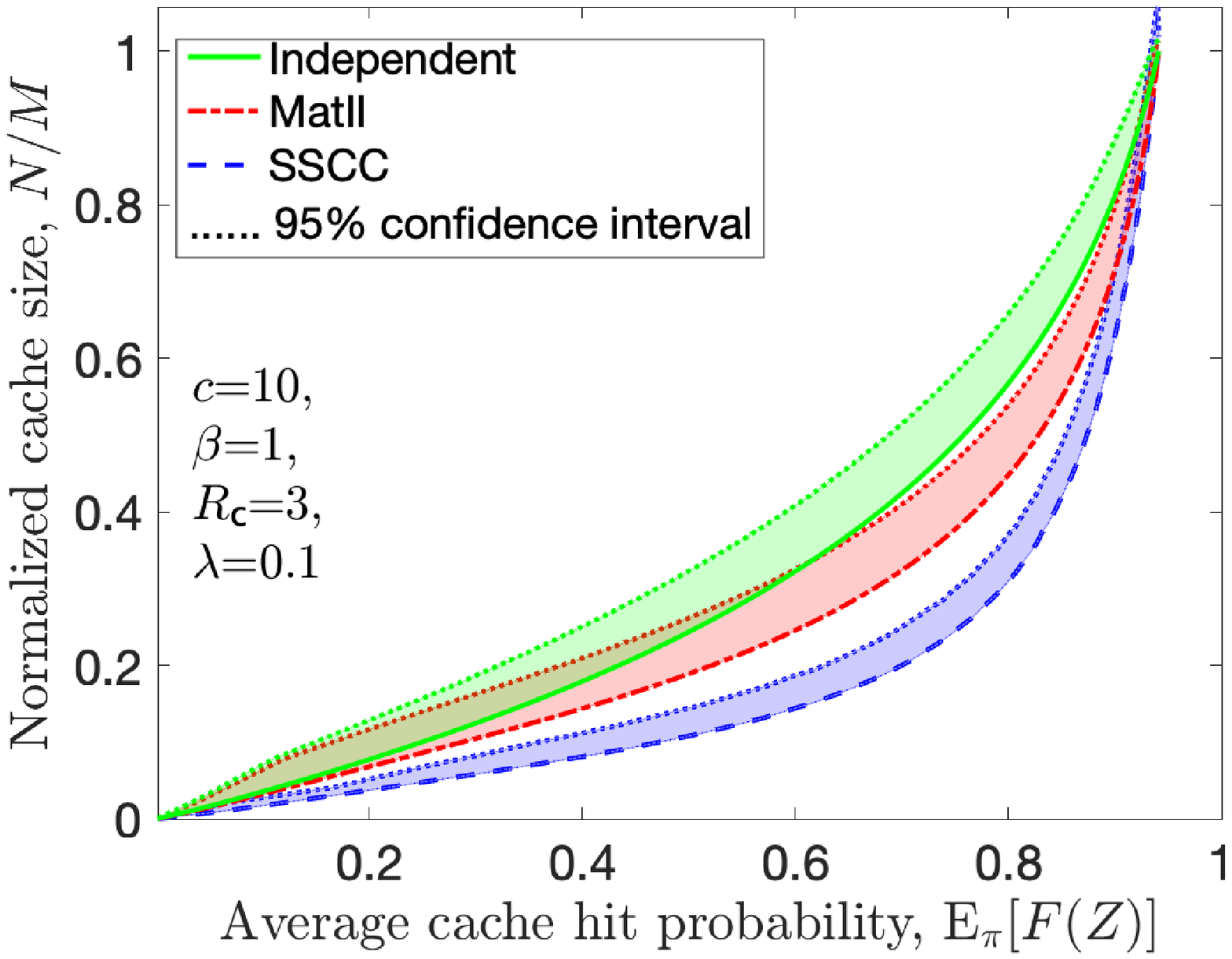}
    \caption{\small{Performance of the non-isotropic model in an urban area scenario. (Left) (\ref{non_uniform_demand_v1}) and (Right) (\ref{non_uniform_demand_v2}). 
    }}
    \label{fig:performance2_nonisotropic}
\end{figure*}

\section{Numerical Simulations}
\label{experiments}

The nodes live in 
a square region of the Euclidean plane with area $L^2$ where $L=100km$. To avoid edge effects, we evaluate the performance only for the middle square region with area $L^2/9$. The network parameters are $\lambda=0.1$ and $\Rdd\in\{3,\,10\}$ which is consistent with the PPP coverage models in \cite{Andrews2011}. We initially assume that the request process is isotropic and Zipf distributed with parameter $\gamma_r=0.1$ over $M=100$ items. 

For MatII, there is a fixed exclusion range for a given item, and we have derived the optimal exclusion radii in \cite{Malak2016twc}. Let $r_i$ be the optimized exclusion range for item $i$ for MatII. For $\GEC$, the weights $v_i\sim U[0,1]$ are i.i.d. for each node $x$ and each item $i$. We assume that the marks $m^{(i)}$ for item $i$ (exclusion radii) are distributed according to a gamma distribution $\mu^{(i)}=\Gamma(0.7 r_i,1)$ for each $x\in\Phi$, and all items $i$, where we choose its parameters such that the average value of the radius mark for item $i$ equals $\bar{m}^{(i)}=0.7 r_i$. Hence, $\Phi_{th}\sim$$\GEC$$[0.1,\Gamma(0.7 r_i,1),U[0,1],1,f]$ with $c=10$, where $c$ is a parameter of $f(r,m,n)$ in (\ref{fcfunction}) that determines its decay rate.   
We can observe that the $\GEC$ model with gamma-distributed exclusion radius can be used to optimize the cache hit probability-cache violation probability tradeoff. However, as variance of exclusion range increases, the violation probability might also increase for a cache hit probability target. We leave the optimization of the distributions of the marks $\mu^{(i)}$ (across all $i$) over a class of distributions and the study of the fundamental performance limits of $\GEC$ as future work.   

We numerically investigate how much cache over-provisioning is required for different spatial cache placement policies: spatially independent \cite{Blaszczyszyn2014}, MatII \cite{Malak2016twc}, and $\GEC$ cache placement. In Fig. \ref{fig:performance1}, we investigate the required cache size $N$ (normalized) of each policy given that the probability of cache violation is small such that $\mathbb{P}[|C(x)-N|\leq \epsilon] > 0.95$ to characterize the required cache size for a given average cache hit probability. We also illustrate the $95\%$ confidence intervals represented by the shaded regions, and mark the cache sizes for different policies when the average cache hit probability is $\mathbb{E}_{\pi}[F(Z)]=0.7$. 
For example, when $\Rdd=3$, for the $95\%$ confidence interval, the excess cache ratio for independent placement in \cite{Blaszczyszyn2014}, and MatII placement in \cite{Malak2016twc} with respect to the $\GEC$ policy is $
142\%$, and $93\%$, respectively. When $\Rdd=10$, the respective excess ratios for the independent and MatII placement policies are $188\%$ and $109\%$, which are illustrated on the plots. $\GEC$ yields a better concentration of the required cache size, which is desired. Hence, it can be exploited to ensure that the cache does not overrun or underrun its capacity constraint. The performance of caching can be improved if the variance of the marks are higher to better exploit the geographic variation of the demand. In this case, the packing will be denser (see Fig. \ref{fig:softcoremodel}), which is desired for a spatially balanced caching design.

We next generate spatially varying request distributions. We use two sets of model parameters for an urban area and a rural area. The parameters are chosen in accordance with the downlink parameters in \cite{zhou2015spatial} such that the logarithm of the peak hour traffic density is characterized by a normal distribution. 
The non-isotropic request (demand) distribution model for urban and rural areas, with the normal distribution parameters ($\mu^* = 15.999$, $\sigma = 1.4116$), and ($\mu^* = 10.2496$, $\sigma = 1.3034$), respectively, are shown in Fig. \ref{fig:nonisotropic_request}. These maps show the peak hour traffic densities in $(\log(Z))({\rm Kbytes/km}^2)$, i.e., the colorbar values are in $\log({\rm Kbytes/km}^2)$ , where the averages are $\sim 10^7$ (${\rm Kbytes/km}^2$) for urban area and $\in[0.8\times 10^4, 2\times 10^4]$ (${\rm Kbytes/km}^2$) for rural area \cite{zhou2015spatial}. The parameter $r$ for the exponential variogram model is chosen to be one third of the range, where ranges for urban and rural areas are $0.0154({\rm km})$ and $ 1.7139({\rm km})$, respectively.

The performance of the non-isotropic request model for (\ref{non_uniform_demand_v1}) and (\ref{non_uniform_demand_v2}) are shown in Fig. \ref{fig:performance2_nonisotropic}. For the non-uniform spatial traffic, we used Fig. \ref{fig:nonisotropic_request} urban area model which contains $120\times 120$ pixels. We randomly picked a $10\times 10$ pixel region and averaged the performance of the spatial request distribution models in (\ref{non_uniform_demand_v1}) and (\ref{non_uniform_demand_v2}), where the request distribution is tilted as function of the peak hour demand intensity. From Fig. \ref{fig:performance2_nonisotropic} (left), we observe that the tilted request model in (\ref{non_uniform_demand_v1}) provides a reduction in the required cache size $N$ for any average cache hit probability desired $\mathbb{E}_{\pi}[F(Z)]$, compared to the uniform request model in Fig. \ref{fig:performance1} (left). Different from the uniform demand model, we see that the scaling for the three different models is similar for $\mathbb{E}_{\pi}[F(Z)]> 0.9$. This is because under this spatially correlated demand model, the tail of the demand is negligible compared to the uniform request model such that less cache over-provisioning is required to capture requests that have lower popularities. While the concentration results are improved because the demand is more correlated rather than being uniform at random, over-provisioning is still required to improve the cache hit performance beyond $\mathbb{E}_{\pi}[F(Z)]>0.8$. Note that for the same amount of provisioning the threshold required for the uniform demand model is approximately $\mathbb{E}_{\pi}[F(Z)]>0.6$, as shown in Fig. \ref{fig:performance1} (left). From Fig. \ref{fig:performance2_nonisotropic} (right), we observe that for the tilted request model in (\ref{non_uniform_demand_v2}), the concentration results are not improved. The performances of the independent and MatII models match the results obtained for the uniform spatial traffic model because the tilting of the demand within a $10\times 10$ pixel region does not modify the demand as the intensity does not fluctuate much within an urban region (see Fig. \ref{fig:nonisotropic_request} (Left)). %If so, this model should be run over larger regions. 
However, the $\GEC$ scheme requires more over-provisioning to achieve a similar average performance.  
The distribution parameters of the gamma-distributed marks could be fine tuned to improve the tail performance of exclusion-based placement configurations.  While more analytical investigation is needed to realize the attainable provisioning performance of the $\GEC$ model under different mark distribution models, note that Gamma distribution is an exponential family that arises in the context of the maximum-entropy distributions with specified means. As a result, the negative dependence is promoted and the performance of GEC is significantly improved, which renders this distribution favorable in the context of heterogeneous networks. 

The achievable caching gains of GEC rely on the assumption that nodes cooperate. Otherwise, each node could perform spatially independent policy \cite{Blaszczyszyn2014}, and obtain local caching gains.  
Because $F(Z)$ is a submodular function, the required cache size to provide the desired guarantee will scale superlinearly in $\mathbb{E}_{\pi}[F(Z)]$. 
For non-uniform demand, the concentration results improve via (\ref{non_uniform_demand_v2}) and (\ref{non_uniform_demand_v1}) that capture the heterogeneity, and the tail of the demand becomes negligible.

Cache size over-provisioning is required for performance optimization in hybrid storage techniques. Our results indicate that the over-provisioning requirement can be made much smaller than the existing spatial caching models, for both the spatially uniform demand and the non-uniform demand. While having no over-provisioned space decreases the performance due to the update cost, having too much over-provisioning can cause performance decrease due to less data caching \cite{oh2012caching}. 
Our $\GEC$ approach shows that it is possible to decrease both the average cache space needed and the over-provisioning required. Via a more effective use of cache capacity, the maintenance costs can be significantly reduced and the lifetime of caches can be prolonged.

\section{Conclusions}
\label{conclusion}
We devised a spatial $\gamma$-exclusion caching model ($\GEC$) in which nodes can cooperate via proximity-based techniques to decide how to populate their caches. Our approach exploits the second order properties of the placement configuration and the variability of the demand in a geographic setting. We demonstrate that our approach provides significantly better concentration of the cache storage size compared to non-cooperative or simple cooperative techniques that do not make a full use of the correlations in content placement. Our results suggest guarantees on the over-provisioned cache space, which can help improve the performance of caching networks. 
We believe that $\GEC$ gives insights into not only how to cache the content, but also how to effectively sample in spatial settings. $\GEC$ is suited for enabling proximity-based applications such as D2D and P2P as it promotes the item diversity and reciprocation. 

Extensions of this work include devising more balanced algorithms exploiting strong Rayleigh measures that imply negative association and have the stochastic covering property \cite[Sect. 4.2]{borcea2009negative}. They also include  
the joint optimization of the spatio-temporal demand dynamics, and employment of exclusion-based models to optimize the performance of time-to-live caches.

\begin{appendix}
\section{Appendices}

\subsection{Functions Describing the Second-Order Behavior of the Point Processes}
\label{SecondOrderProperties}
The spatial regularity and second-order properties of $\Phi$ can be characterized by the reduced second moment function which is known as Ripley's $K$-function $K(r)$, $r\geq 0$. The mean number of points of $\Phi$ within a ball of radius $r$ and centred at the typical point, which is not itself counted in the mean is given by $\lambda K(r) = \mathbb{E}^{!\circ}[\Phi(B_0(r))]$ \cite[Ch. 2.3]{Stoyan1996}. While the $K$ function for $\GEC$ is not easy to characterize, some special cases are easy to analyze.

Using the Campbell's theorem \cite[Ch. 1.4]{BaccelliBook1}, we deduce that the average number of transmitters of the stationary point process $\Phi_{th,i}$ -- conditioned on there being a point at the origin but not counting it -- contained in the ball $B=B_0(\Rdd)$ equals
\begin{align}
\label{reducedPalmMeasure}
\lambda K(B_0(\Rdd)) &= \mathbb{E}^{!\circ}[\Phi_{th,i}(B)]\\
&=\lambda^{-1}\int\nolimits_{B_0(\Rdd)}\rho_i^{(2)}(x){\rm d}x=\lambda \int\nolimits_{B_0(\Rdd)} k_i(x)\,{\rm d}x,\nonumber
\end{align}
where $\rho^{(2)}_i (r)$ is the second-order product density (SOPD) corresponding to $\Phi_{th,i}$, and $k_i(r)=\lambda^{-2} \rho^{(2)}_i (r)$ is the two-point Palm probability that two points of $\Phi$ separated by distance $r$ are both retained to store item $i$. For a stationary p.p. $\Phi_{th}$, the SOPD is the joint probability that there are two points of $\Phi_{th}$ at locations $x$, $y$ in the infinitesimal volumes $dx$, $dy$ \cite[Ch. 5.4]{Stoyan1996}. 

We next give the SOPD for MatII.
\begin{prop}
For a stationary p.p. $\Phi_{th}$ that is MatII, the SOPD is given by \cite[Ch. 5.4]{Stoyan1996}
\begin{align}
\label{SOPD_MatII}
   & \rho^{(2)}_i(r) =\lambda^2_{\rm hcp}\times \nonumber\\
   &\begin{cases}
    1, \quad r\geq 2r_i\\
    \frac{2V_{r_i}(r)[1-\exp(-\lambda\pi r_i^2)]-2\pi r_i^2 [1-\exp(-\lambda V_{r_i}(r))]}{\pi r_i^2 V_{r_i}(r)[V_{r_i}(r)-\pi r_i^2]},\quad r_i<r<2r_i,\\
    0,\quad r\leq r_i,
    \end{cases}
\end{align}
where $V_{r_i}(r) = 2\pi r^2_i- l_2(r_i,r)=2\pi r^2_i - 2r^2_i \cos^{-1}\big(\frac{r}{2r_i}\big)+r\sqrt{r_i^2 -\frac{r^2}{4}} $ \cite[Eqn. (5.60)]{Stoyan1996} is the area of the union of two circles with radii $r_i$ that are separated by $r$ where $l_2(r,\delta)$ is given in (\ref{AreaIntersectingCircles}). Pairwise correlations between the points separated by $r > r_i$ are modeled using the SOPD -- $\rho^{(2)}_i (r)$ for item $i$ -- of the MatII hard-core point process.
\end{prop}

A special case of MatII is MatI in which all nodes have the same mark values. The $K$ function for MatI can be given using (\ref{reducedPalmMeasure}), where a pair of points of $\Phi$ separated by less than a critical distance  
$r_i$ are deleted and thus $\lambda_{\rm hcp-I}=\lambda e^{-\lambda\pi r_i^2}$. 
The second-order product density for MatI is
\begin{align}
\label{k_MatI}
\rho_i^{(2)}(r)=\begin{cases}
\lambda^2\exp(-\lambda (2\pi r_i^2-V_{r_i}(r))),\quad r\geq r_i,\\
0,\quad r<r_i.
\end{cases}
\end{align}

The covariance or two-point probability function of a random set with volume fraction $p$ is
\begin{align}
C(r)&=\mathbb{P}(0\in B,\, {\bf r}\in B)\nonumber\\
&=\mathbb{P}(0\notin B,\, {\bf r}\notin B)+2p-1,\quad {\bf r}\in\mathbb{R}^2. 
\end{align}
For example, if $\Xi$ is a Boolean model with typical grain $\Xi_0$, then its covariance is
\begin{align}
C(r)=2p-1+(1-p)^2\exp(\lambda \mathbb{E}[\gamma_{\Xi_0}(r)]),
\end{align}
where $\gamma_{B}(r)=\nu_2(B \cap (B-{\bf r}))$ is the set covariance of the convex set $B$, where ${\bf r}\in\mathbb{R}^2$ with $\norm{{\bf r}}=r$, and $\nu_2$ is the 2-dimensional Lebesgue measure (the area measure) \cite[1.7.2]{Stoyan1996}.

We summarize the functions for the baseline PPP $\Phi$ and MatII $\Phi_{th,i}$ models in Table \ref{functiontable}.

\begin{table*}[t!]\footnotesize
\centering
\setlength{\extrarowheight}{4pt}
\begin{tabular}{| l | l | l | l | l |}
\hline
{\bf Definition} & {\bf Function} & {\bf Poisson p.p.} $\Phi$ & {\bf Mat\'{e}rn II hard-core p.p.} $\Phi_{th,i}$\\
\hline
Volume fraction & $p$ & $p=1$ & $p=\lambda_{\rm hcp}(i)/\lambda$ \\ 
\hline
Covariance
& $C(r)=C_B(r)$ & $p^2$ & $\approx p(1-p)\exp{(-\alpha r)}+p^2$ \cite{debye1957scattering}, \cite{bondesson2003mean}\\ 
\hline
Set covariance (Variogram) & $\gamma_B(r)=C(0)-C(r)$ & $0$ & $p-C(r)$ \\ 
\hline
Covariance function & $k(r)=C(r)-p^2$ & $1$ & $\lambda^{-2}\,\rho^{(2)}_i(r)$ \\ %5.10
\hline
Ripley's $K$ function & $K(r)$ & $\pi r^2$ & $\lambda^{-1} \mathbb{E}^{!\circ}[\Phi_{th,i}(B)]$ 
\\
\hline
Second-order product density & $\rho^{(2)}(r)=\lambda^2\frac{{\rm d}K(r)}{{\rm d}r}\Big/2\pi r$ & $\lambda^2$ & Eq. (\ref{SOPD_MatII}).\\
\hline
Pair correlation function & $g(r)=\frac{\rho^{(2)}(r)}{\lambda^2}\approx\frac{C(r)}{p^2}$ & 1 &  Eq. $(\ref{SOPD_MatII})/\lambda^2$ \\
\hline
\end{tabular}
\caption{Functions describing the second-order behavior of the point processes in $\mathbb{R}^2$ \cite[Ch. 4]{Stoyan1996}. For the exponential covariance approximation for MatII given above, $\alpha\in\mathbb{R}^+$ describes the degree of variability 
in $B$ 
\cite{debye1957scattering}.}
\label{functiontable}
\end{table*}

\subsection{Proof of Proposition \ref{var_neg_assoc}}
\label{App:var_neg_assoc}
We can compute the variance of $F(Z)=\sum\limits_{(i,q)\in\mathcal{R}}{\lambda_{(i,q)}\sum\limits_{k=1}^{|q|-1}w_{q_{k+1}q_k}{\Big(1-f_{z_i}(1,\hdots, k)\Big)} }$ as
\begin{align}
&\mathrm{Var}[F(Z)]=\mathbb{E}\Big[\Big(F(Z)\Big)^2\Big]-\Ex{F(Z)}^2\nonumber\\
&=\sum\limits_{(i,q)\in\mathcal{R}}\sum\limits_{(i',q')\in\mathcal{R}}\sum\limits_{k=1}^{|q|-1}\sum\limits_{k'=1}^{|q'|-1}\lambda_{(i,q)}\lambda_{(i',q')}w_{q_{k+1}q_k}w_{q'_{k'+1}q'_{k'}}\times\nonumber\\
&\big\{\Ex{f_{z_i}(1,\hdots, k))f_{z_{i'}}(1,\hdots, k')}\nonumber\\
&-\Ex{f_{z_i}(1,\hdots, k)}\Ex{f_{z_{i'}}(1,\hdots, k')}\big\}\nonumber\\
&=\sum\limits_{(i,q)\in\mathcal{R}}\sum\limits_{k=1}^{|q|-1}\lambda_{(i,q)}^2w_{q_{k+1}q_k}^2 \big\{\Ex{f_{z_i}(1,\hdots, k)^2}\nonumber\\
&-\Ex{f_{z_i}(1,\hdots, k)}^2  \big\} \nonumber\\
&+\sum\limits_{(i,q)\in\mathcal{R}}\sum\limits_{\underset{(i',q')\neq (i,q)}{(i',q')\in\mathcal{R},}}\sum\limits_{k=1}^{|q|-1}\sum\limits_{k'=1}^{|q'|-1}\lambda_{(i,q)}\lambda_{(i',q')}w_{q_{k+1}q_k}w_{q'_{k'+1}q'_{k'}}\times\nonumber\\
&\big\{\Ex{f_{z_i}(1,\hdots, k)f_{z_{i'}}(1,\hdots, k')}\nonumber\\
&-\Ex{f_{z_i}(1,\hdots, k)}\Ex{f_{z_{i'}}(1,\hdots, k')}\big\}.\nonumber
\end{align}
In the above expression, the terms $\{\Ex{f_{z_i}(1,\hdots, k))f_{z_{i'}}(1,\hdots, k')}\}$ can be evaluated using the first-order and second-order properties of $f_{z_i}(1,\hdots, k)$ in (\ref{product_observation}) 
\begin{align}
\mathbb{E}\Big[f_{z_i}(1,\hdots, k) f_{z_j}(1,\hdots, k)\Big]=\mathbb{P}\Big(\sum\limits_{k=i,j} \Phi_{th,k}(B)=0\Big)\nonumber\\
=\mathbb{P}\Big(\sum\limits_{k'=1}^k z_{q_{k'} i}=0,\,\sum\limits_{l'=1}^l z_{q_{l'} j}=0\Big),\nonumber
\end{align}
where the first order properties of $f_{z_i}(1,\hdots, k)$ are derived in Prop. \ref{NA}. Similarly, we have 
\begin{align}
\Cov{f_{z_i}(1,\hdots, k),\,f_{z_j}(1,\hdots, l)}=\mathbb{P}\Big(\sum\limits_{k=i,j} \Phi_{th,k}(B)=0\Big)\nonumber\\
-\mathbb{P}\Big(\sum\limits_{k'=1}^k z_{q_{k'} i}=0\Big) \mathbb{P}\Big(\sum\limits_{l'=1}^l z_{q_{l'} j}=0\Big).\nonumber
\end{align}
We note that the variance of $F(Z^*)$ is given as
\begin{align}
\mathrm{Var}[F(Z^*)]=\sum\limits_{(i,q)\in\mathcal{R}}\sum\limits_{k=1}^{|q|-1}\lambda_{(i,q)}^2w_{q_{k+1}q_k}^2 \Big\{\prod\limits_{l=1}^k\Ex{\big(1-z^*_{q_l i}\big)^2}\nonumber\\
-\prod\limits_{l=1}^{k}\Ex{\big(1-z^*_{q_{l} i}\big)}^2  \Big\}.\nonumber
\end{align}
Given feasible negatively associated sequences $Z$, we have %that 
\begin{multline}
\big\{\Ex{f_{z_i}(1,\hdots, k)f_{z_{i'}}(1,\hdots, k')}\nonumber\\
-\Ex{f_{z_i}(1,\hdots, k)}\Ex{f_{z_{i'}}(1,\hdots, k')}\big\} \leq 0,\nonumber
\end{multline}
and equality is satisfied for $Z=Z^*$. 

Our goal is to prove that the following relation holds: 
\begin{align}
\Ex{f_{z_i}(1,\hdots, k)^2}-\Ex{f_{z_i}(1,\hdots, k)}^2 \nonumber\\
\leq \prod\limits_{l=1}^k\Ex{\big(1-z^*_{q_l i}\big)^2}-\prod\limits_{l=1}^{k}\Ex{\big(1-z^*_{q_{l} i}\big)}^2.\nonumber 
\end{align}
Since $\big(1-z_{q_l i}\big)^2=\big(1-z_{q_l i}\big)$ for binary variables, we need to show that 
\begin{align}
&\Ex{f_{z_i}(1,\hdots, k)}-\prod\limits_{l=1}^k\Ex{\big(1-z^*_{q_l i}\big)} \nonumber\\
&\leq \Ex{f_{z_i}(1,\hdots, k)}^2-\prod\limits_{l=1}^{k}\Ex{\big(1-z^*_{q_{l} i}\big)}^2 \nonumber\\
&\leq \Big(\Ex{f_{z_i}(1,\hdots, k)}-\prod\limits_{l=1}^{k}\Ex{\big(1-z^*_{q_{l} i}\big)}\Big)\nonumber\\
&\times\Big(\Ex{f_{z_i}(1,\hdots, k)}+\prod\limits_{l=1}^{k}\Ex{\big(1-z^*_{q_{l} i}\big)}\Big).\nonumber
\end{align}
Since $\Ex{f_{z_i}(1,\hdots, k)}-\prod\nolimits_{l=1}^k\Ex{\big(1-z^*_{q_l i}\big)}\leq 0$, the above condition is true when $\Ex{f_{z_i}(1,\hdots, k)}+\prod\nolimits_{l=1}^{k}\Ex{\big(1-z^*_{q_{l} i}\big)}\in[0,1]$. Furthermore, because $\Ex{f_{z_i}(1,\hdots, k)}\leq \prod\nolimits_{l=1}^{k}\Ex{\big(1-z^*_{q_{l} i}\big)}$, it is sufficient to have $ \prod\nolimits_{l=1}^{k}\Ex{\big(1-z^*_{q_{l} i}\big)}\in[0,1/2]$. This is true whenever $\Ex{z^*_{q_{l} i}}\leq 1/2$ for all $q_l\in q$.

\subsection{Proof of Prop. \ref{Avg_Hit_MHCP}}
\label{App:Avg_Hit_MHCP}

For the special case of MatII (Corollary \ref{MatIIsimplified}), we have that
\begin{align}
\label{equalityfortail}
\mathbb{P}(\Phi_{th,i}(B)>0)=\lambda_{\rm hcp}(i)\pi\Rdds,\,\, r_i\geq\Rdd.
\end{align}
The cache hit probability  for $r_i<\Rdd$ is bounded below as
\begin{align}
\label{LB}
\mathbb{P}(\Phi_{th,i}(B)>0)\geq 1-\exp(-\lambda_{\rm hcp}(i)\pi\Rdds),\,\, r_i<\Rdd,
\end{align}
which is due to the negative association property of MatII. Hence, the miss probability is lower than that of the independently thinning of the original PPP $\Phi$  
with probability $\lambda_{\rm hcp}(i)/\lambda$.

Similarly, the cache hit probability  for $r_i<\Rdd$ can be bounded above as
\begin{align}
\label{UB}
\mathbb{P}(\Phi_{th,i}(B)>0)&\leq 1-\exp(-\lambda_{\rm hcp}(i)\pi\Rdds)\nonumber\\
&+\lambda^{-1}\int\nolimits_{B_0(\Rdd)}\rho_i^{(2)}(x){\rm d}x,\quad r_i<\Rdd,
\end{align}
where the upper bound follows from Markov's inequality.

\subsection{Proof of Prop. \ref{MHCP}}
\label{VarMHCII}

Using the notation $A_i=\mathbbm{1}(\Phi_{th,i}(B)>0)$ for $i=1,\hdots, M$, the variance of $F(Z)$ satisfies
\begin{align}
&\mathrm{Var}[F(Z)]
=\sum\limits_{i}{p_r^2(i)\mathrm{Var}[A_i]}\nonumber\\
&+2\sum\limits_{1\leq i<j\leq M}{p_r(i)p_r(j)\Cov{A_i,A_j}}
\overset{(a)}{=}\sum\limits_{i}{p_r^2(i)\mathrm{Var}[A_i]}\nonumber\\
&=\sum\limits_{i}{p_r^2(i)\mathbb{P}(\Phi_{th,i}(B)>0)(1-\mathbb{P}(\Phi_{th,i}(B)>0))}\nonumber\\
&\overset{(b)}{\leq}\sum\limits_{i=1}^{m_c}{p_r^2(i)\Big[\lambda_{\rm hcp}(i)\pi\Rdds+\lambda^{-1}\int\nolimits_{B_0(\Rdd)}\rho_i^{(2)}(x){\rm d}x\Big]e^{-\lambda_{\rm hcp}(i)\pi\Rdds}}\nonumber\\
&+\sum\limits_{i=m_c+1}^{M}{ p_r^2(i)\lambda_{\rm hcp}(i)\pi\Rdds e^{-\lambda_{\rm hcp}(i)\pi\Rdds}}\nonumber\\
&=\sum\limits_{i=1}^{M}{p_r^2(i)\lambda_{\rm hcp}(i)\pi\Rdds e^{-\lambda_{\rm hcp}(i)\pi\Rdds}}\nonumber\\
&+\sum\limits_{i=1}^{m_c}{p_r^2(i)\Big[\lambda^{-1}\int\nolimits_{B_0(\Rdd)}\rho_i^{(2)}(x){\rm d}x\Big]e^{-\lambda_{\rm hcp}(i)\pi\Rdds}}\nonumber\\
&\leq \frac{1}{e}\sum\limits_{i=1}^{M}{p_r^2(i)}+\frac{2\pi}{\lambda}\sum\limits_{i=1}^{m_c}{p_r^2(i)\int\nolimits_{r_i}^{\Rdd}\rho_i^{(2)}(r)r{\rm d}r}\nonumber\\
&\leq \frac{1}{e}\sum\limits_{i=1}^{M}{p_r^2(i)}+\frac{\pi}{\lambda}\sum\limits_{i=1}^{m_c}{p_r^2(i)\frac{\big(1-e^{-\lambda \pi r_i^2}\big)^2}{\pi^2 r_i^4}(\Rdds-r_i^2)}\nonumber\\
& = \frac{1}{e}\sum\limits_{i=1}^{M}{p_r^2(i)}+\pi\lambda\sum\limits_{i=1}^{m_c}{p_r^2(i)p_c^2(i)(\Rdds-r_i^2)}\nonumber\\
&\leq \sum\limits_{i=1}^{M}{p_r^2(i)\Big(\frac{1}{e}+\pi\lambda p_c^2(i)(\Rdds-r_i^2)_+\Big)},\nonumber
\end{align}
where $(a)$ is due to the independent placement of each item across caches. $(b)$ is from (\ref{LB})-(\ref{UB}) and by noting the relation $1-\exp(-\lambda_{\rm hcp}(i)\pi\Rdds)\leq \lambda_{\rm hcp}(i)\pi\Rdds$.

\subsection{Proof of Theorem \ref{convexity}}
\label{App:convexity}

From Prop. \ref{CTPP_SSCC}, we have that
\begin{align}
\label{eta_SSCC}
\etaGM(r,\vec{\delta})&=\int\nolimits_{\mathbb{R}}\int\nolimits_{0}^1
    \exp\Big(-u \lambda \int\nolimits_{\mathbb{R}} (\pi (m+n)^2\nonumber\\
    &-l_2(r,n)) \, 
    \mu({\rm d}n) \Big)\, {\rm d}u\, \mu({\rm d}m) \nonumber\\
    &=\mathbb{E}_{m}\Big[\mathbb{E}_{U}\Big[
    \exp{\left(-U q(\lambda,r,m)\right)}\Big]\Big] \nonumber\\
    &=\mathbb{E}_{m}\left[\frac{1-\exp{\left(-q(\lambda,r,m)\right)}}{q(\lambda,r,m)}\right].
\end{align}
In (\ref{eta_SSCC}), $U\sim U[0,1]$ is a uniformly distributed random variable, and
\begin{align}
q(\lambda,r,m)&=\lambda \mathbb{E}_{m_2}\Big[\pi (m+m_2)^2-l_2(r,m_2)\Big]\nonumber\\
&=\lambda\pi \big(m^2+2m\bar{m}_2\big)+\lambda \mathbb{E}_{m_2}\big[\pi m_2^2-l_2(r,m_2)\big].\nonumber
\end{align}
In (\ref{eta_SSCC}), let $f(m)=\exp(-U\lambda\pi(m^2+2m\bar{m}_2))$. Hence, $f'=\big(-U\lambda \pi(2m+2\bar{m}_2)\big)f<0$, and $f''=(-2U\lambda\pi)f+\big(U \lambda \pi\big(2m + 2\bar{m}_2 \big)\big)^2f\geq 0$ that is convex in $m$ provided that $m + \bar{m}_2\geq \sqrt{\frac{1}{2U\lambda \pi}}$. Since the baseline process $\Phi$ is a homogeneous PPP with intensity $\lambda$, we have $\bar{m}_2\geq\frac{1}{2\sqrt{\lambda}}$. If $\bar{m}+m_2\geq\frac{1}{\sqrt{\lambda}}$ is satisfied, then the condition $U>\frac{1}{2\pi}$ is sufficient for the convexity. 
In addition, we can show that $\etaGM(r,\vec{\delta})=\mathbb{E}_{m}[g(m)]$, with $g=\frac{f-1}{\log(f)}=\frac{1-\exp{(-x)}}{x}$ where $x=U\lambda\pi(m^2+2m\bar{m}_2)$. Then $g'=\frac{\exp{(-x)}(x+1)-1}{x^2}$, and 
$g''=\frac{2-\exp{(-x)}[x^2+2x+2]}{x^3}>0$. 
Hence, $$\etaGM(r,\vec{\delta})=\mathbb{E}_{m}[g(m)]\geq g(\bar{m})=\etaM(r,\bar{m}).$$

\subsection{Proof of Prop. \ref{spatial_var}}
\label{App:spatial_var}
From \cite[Ch. 4]{Stoyan1996}, the variances of point-counts of $\Phi_{th,i}$ can be calculated as  
\begin{align}
\mathrm{Var}[\Phi_{th,i}(B)]&=\lambda^2_{\rm hcp}(i)\int\nolimits_0^{\infty} \gamma_B(r){\rm d}K(r)\nonumber\\
&+\lambda_{\rm hcp}(i) \pi \Rdds - (\lambda_{\rm hcp}(i) \pi \Rdds)^2,\nonumber 
\end{align}
where $\gamma_B(r)$ is the set covariance of the convex set $B$ for MatII (see Appendix \ref{SecondOrderProperties}, Table \ref{functiontable}). 

For large $B=B_0(\Rdd)$, i.e., spherically infinite according to \cite{girling1982approximate}, we have that
\begin{align}
    &\frac{\mathrm{Var}[\Phi_{th,i}(B)]}{\pi \Rdds}\approx \lambda_{\rm hcp}(i)+2\pi \lambda^2_{\rm hcp}(i) \int\nolimits_0^{\infty} (g_i(r)-1)r {\rm d}r \nonumber\\
    &= \lambda_{\rm hcp}(i)-4\pi\lambda^2_{\rm hcp}(i) r_i^2+2\pi\int_{r_i}^{2r_i} \rho_i^{(2)}(r)r{\rm d}r\nonumber\\
    &\overset{(a)}{\leq}\lambda_{\rm hcp}(i)-\pi\lambda^2_{\rm hcp}(i) r_i^2 
    \overset{(b)}{\leq} \lambda_{\rm hcp}(i) e^{-\lambda\pi r_i^2}, \nonumber
\end{align}
where $g_i(r)$ is the pair correlation function of $\Phi_{th,i}$ (see Table \ref{functiontable}), $(a)$ follows from observing that $\rho_i^{(2)}(r)\leq \lambda^2_{\rm hcp}(i)$ from (\ref{SOPD_MatII}), and $(b)$ from $1-x-\exp(-x)\leq 0$ for $x\geq 0$. 

\end{appendix}

%\begin{spacing}{0.82}
\bibliographystyle{IEEEtran}
\bibliography{references}
%\end{spacing}

\end{document}